\documentclass[onecolumn]{IEEEtran}
\usepackage[normalem]{ulem}
\usepackage[utf8]{inputenc} %
\usepackage[english]{babel}
\usepackage{mathrsfs}
\usepackage[dvipsnames]{xcolor}
\usepackage{%
algorithm,%
algpseudocode,%
amsfonts,%
amsmath,%
amssymb,%
amsthm,%
array,%
bbm,%
blindtext,%
cite,%
comment,%
graphicx,%
hyperref,%
paralist,%
pgf,%
pgfgantt,%
pgfplots,%
booktabs,%
tabularx,%
textcomp,%
tikz,%
times,%
titlesec,%
url,%
xcolor%
}

\usepackage[mathscr]{eucal}
\usepackage[caption=false,font=footnotesize]{subfig}
\graphicspath{{Figures/}{../Figures/}}

\usetikzlibrary{
arrows,
calc,
positioning
}

\renewcommand{\le}{\leqslant}
\renewcommand{\ge}{\geqslant}

\newcommand{\EQ}[1]{\begin{equation*}#1\end{equation*}}
\newcommand{\EQN}[1]{\begin{equation}#1\end{equation}}
\newcommand{\eq}[1]{\begin{align*}#1\end{align*}}

\newcommand{\meq}[2]{\begin{xalignat*}{#1}#2\end{xalignat*}}

\newcommand{\set}[1]{\left\{#1\right\}}

\newcommand{\SetIn}[1]{\mathbbm{1}_{\set{#1}}}
\newcommand{\abs}[1]{\left\lvert #1\right\rvert}

\newcommand{\twocgs}{\mathrm{2CGS}}
\newcommand{\Bern}{\mathrm{Bern}}
\newcommand{\Bin}{\mathrm{Bin}}
\newcommand{\Pois}{\mathrm{Pois}}
\newcommand{\iid}{\emph{i.i.d.~}}

\newcommand{\TV}{\mathrm{TV}}

\newcommand{\meandeg}{\E D}

\newcommand{\E}{\mathbb{E}}

\newcommand{\N}{\mathbb{N}}
\newcommand{\R}{\mathbb{R}}
\newcommand{\Z}{\mathbb{Z}}

\newcommand{\nbrh}{\mathcal{N}}

\newcommand{\cM}{\mathcal{M}}

\newcommand{\hJ}{\hat{J}}

\newcommand{\hpi}{\hat{\pi}}

\newcommand{\sF}{\mathscr{F}}

\newcommand{\sX}{\mathscr{X}}

\newcommand{\tE}{\tilde{E}}

\newcommand{\tJ}{\tilde{J}}

\newcommand{\tpi}{\tilde{\pi}}

\newcommand{\Ind}[1]{\mathbbm{1}_{#1}}

\newcommand{\red}[1]{{\color{red}#1}}

\newcommand{\nix}[1]{}

\theoremstyle{plain}
\newtheorem{thm}{Theorem}
\newtheorem{cor}[]{Corollary}
\newtheorem{lem}[]{Lemma}

\newtheorem{probl}[]{Problem}

\theoremstyle{definition}
\newtheorem{defn}[]{Definition}

\newtheorem{assum}[]{Assumption}

\theoremstyle{remark}
\newtheorem{rem}{Remark}

\newtheorem{innercustomgeneric}{\customgenericname}
\providecommand{\customgenericname}{}
\newcommand{\newcustomtheorem}[2]{%
\newenvironment{#1}[1]
{%
 \renewcommand\customgenericname{#2}%
 \renewcommand\theinnercustomgeneric{##1}%
 \innercustomgeneric
}
{\endinnercustomgeneric}
}

\newcustomtheorem{conj2}{Conjecture}
\newcustomtheorem{exmp}{Example}

\begin{document}

\title{Bipartite matching under communication constraints}

\author{
Moonmoon Mohanty\IEEEauthorrefmark{1}, \IEEEmembership{Student Member, IEEE},
\and Gautham Bolar\IEEEauthorrefmark{1},
\and Preetam Patil\IEEEauthorrefmark{1},
\and Ayalvadi Ganesh\IEEEauthorrefmark{2},
\and Jean-Francois Chamberland\IEEEauthorrefmark{3}, \IEEEmembership{Senior Member, IEEE},
\and Parimal Parag\IEEEauthorrefmark{1}, \IEEEmembership{Senior Member, IEEE}
\thanks{
Authors\IEEEauthorrefmark{1} are with the Department of Electrical Communication Engineering and the Centre for Networked Intelligence, Indian Institute of Science, Bangalore, Karnataka 560012, India.
Author\IEEEauthorrefmark{2} is with the School of Mathematics, University of Bristol, Bristol BS8 1UG, UK.
Author\IEEEauthorrefmark{3} is with the Electrical and Computer Engineering, Texas A\&M University, College Station, TX 77843, USA.
Email: \{moonmoonm, gauthamb, preetampatil, parimal\}@iisc.ac.in, a.ganesh@bristol.ac.uk, chmbrlnd@tamu.edu.
}
}
\maketitle
\begin{abstract}
In modern data center networks, thousands of hosts contend for shared link capacity; the scale of these systems makes centralized scheduling impractical.
This article models such scheduling as a bipartite matching problem under communication constraints: senders express interest in forming connections, and receivers respond using only locally available information.
A class of single-round probabilistic matching algorithms is proposed, built on two key ideas: degree-biased sampling, where senders use receiver degrees to inform their random selection; and random thinning, where senders report only a random subset of their connections.
Analytical performance guarantees are established for random graph models. 
In sparse regimes, degree-biased sampling yields a higher expected matching size than prior communication-constrained algorithms; in denser settings, a counterintuitive phenomenon emerges where deliberately restricting available connections through thinning increases the expected number of matches.
Combining thinning to degree two with greedy selection produces an algorithm that requires no parameter tuning and, in packet-level simulations with production traffic traces, significantly extends the network stability region.
Although motivated by data center network scheduling, the underlying framework of bipartite matching under local information constraints is portable to other resource allocation settings.
\end{abstract}

\section{Introduction}

Resource allocation and scheduling in large-scale computing clusters increasingly depend on distributed decision-making.
In modern data center networks (DCNs), thousands of hosts equipped with graphics processing units (GPUs) exchange data across high-bandwidth, multi-path fabrics, yet the sheer scale of these systems makes centralized coordination impractical.
This mismatch between the need for efficient resource assignment and the constraints on information exchange motivates the study of matching-based allocation under limited communication.

Within DCNs, traffic is heterogeneous: short coordination flows demand minimal latency, while massively parallel workloads, such as distributed training and large-scale data processing, generate bulk transfers that require sustained throughput \cite{dctcp_sigcomm2010, d2tcp_vamanan_sigcomm2012, das_tcp_comsnets2013, fastpass_sigcomm14}.
At scale, these competing demands strain receiver-side resources, particularly when many senders target the same host.
Traditional reactive congestion control, which relies on delay or packet-drop feedback, cannot resolve such contention efficiently~\cite{ousterhout2023itstimereplacetcp}.

As device density continues to grow, centralized solutions that were once practical are becoming untenable, driving a shift toward distributed resource-allocation frameworks where local decisions must be made with only partial information.
Motivated by these trends, we consider settings in which resource allocation can be abstracted as a bipartite matching problem under communication constraints.

The literature on bipartite matching is extensive, featuring classic approaches such as the Hungarian algorithm for assignment problems and the Hopcroft--Karp algorithm for maximum matchings \cite{kuhn1955hungarian, ford1957simple, munkres1957algorithms, hopcroft1973n}.
These algorithms require complete knowledge of the graph; in particular, the augmenting-path computations central to Hopcroft--Karp cannot be carried out with only local information.
When full knowledge is unavailable, different procedures must be adopted, often at the expense of achieving a maximum matching.

While several threads of related work address aspects of distributed matching, from randomized load balancing to multi-round graph algorithms to online assignment, the specific regime of single-round structural matching under communication constraints has received comparatively little attention.
Existing work on distributed matching often assumes agents hold preference orderings or value functions over potential partners \cite{dobzinski2014economic,anshelevich2019tradeoffs,assadi2021auction}, whereas in DCN scheduling the problem is purely structural: edges represent feasible sender--receiver pairs, and the objective is to find a large matching using only local information exchange.
This article aims to bridge this gap by introducing a class of single-round bipartite matching algorithms with performance guarantees for random graph models representative of DCN traffic.

Specifically, nodes possess only local knowledge of the graph and exchange minimal coordination messages.
Senders declare their intentions to connect, receivers respond with local degree information, and senders then select a receiver through a probabilistic sampling process biased by the degrees of adjacent receivers.
This degree-biased approach yields a higher mean number of matches than existing communication-constrained algorithms, particularly in sparse-traffic regimes \cite{pim_anderson1993high,dcpim_sigcomm2022}.

In denser settings, an interesting phenomenon arises.
With centralized knowledge, additional edges in the connectivity graph promote a larger maximum matching.
However, when only local information is available, random thinning, where senders report only a random subset of their connections, can actually increase the expected number of matches.
We investigate this phenomenon through parameter optimization and empirical study, with potential applications to dense DCN workloads such as distributed training \cite{wang2024railonlyllm}, web search \cite{dean2008mapreduce,dctcp_sigcomm2010}, and large-scale data processing \cite{vl2_sigcomm2009,shvachko2010hadoop, zaharia2016apache}.

Multiple-round matching protocols can increase the expected matching size \cite{pim_anderson1993high, islip_tnet1999, dcpim_sigcomm2022}, but each additional round incurs control-message overhead that burdens latency-sensitive short flows and, in some designs, introduces nontrivial matching delays.
In practice, real-time systems cannot afford such overheads, making fast, single-round algorithms far more attractive for deployment.

\subsection{Related Work}
\label{sec:literature}

The present work draws on three threads of research: randomized load balancing through restricted choices, distributed matching under locality constraints, and matching-based scheduling in data center networks (DCNs).
We survey each in turn and explain how they inform the algorithms proposed in this article.

A foundational insight in resource allocation is the \emph{power of two choices}: when placing $n$ balls into $n$ bins, selecting the less loaded of two randomly chosen bins reduces the maximum load from $\Theta(\log n / \log\log n)$ to $\Theta(\log\log n)$, an exponential improvement \cite{azar1999balanced, mitzenmacher2001power}.
This result demonstrates that even minimal local information can dramatically improve the outcome of a random allocation process.
Our degree-biased sampling mechanism is motivated by the same principle: senders use the degrees of adjacent receivers to bias their random selection, analogous to choosing the less loaded bin.
Similarly, random thinning, in which senders report only a random subset of their connections, can be viewed as a deliberate restriction of the choice set to facilitate rapid decentralized decisions.
There is, however, a fundamental structural difference: the balls-into-bins framework is inherently sequential and persistent (balls arrive one at a time and bins accumulate load), whereas our setting requires a one-shot combinatorial matching in which all senders act simultaneously and each receiver can be matched to at most one sender.
This distinction means that the quantitative gains from limited choices do not transfer directly, but the qualitative lesson that a small amount of local information substantially outperforms purely random allocation remains central to our approach.

As noted above, classical matching algorithms require complete knowledge of the graph.
When only local communication is available, a natural question is what can be achieved in a bounded number of synchronous rounds, where each node exchanges messages only with its immediate neighbors.
In this setting, Israeli and Itai \cite{israeli1986fast} show that a maximal matching can be computed in $O(\log n)$ randomized rounds, and Linial \cite{linial1992locality} establishes that even basic symmetry-breaking tasks require $\Omega(\log^* n)$ rounds.
Computing a maximum (as opposed to maximal) matching is harder still, as augmenting paths may span the entire graph.
These results delineate a fundamental tension between the locality of information and the quality of the matching.
In the online setting, Karp, Vazirani, and Vazirani \cite{karp1990optimal} establish a tight competitive ratio of $1 - 1/e$ for bipartite matching when one side arrives sequentially and decisions are irrevocable.
Although our model is not online in the adversarial sense, it shares the theme of irrevocable commitment under incomplete information: each sender must choose a single receiver based on local data, without knowledge of other senders' decisions.
Our single-round constraint is more restrictive than what these models typically consider, making performance guarantees in this regime particularly meaningful.

Turning to DCN transport protocols, in \cite{pfabric} Alizadeh et al.\ propose pFabric, which provides near-optimal message/flow\footnote{`Flow' and `message' have been used interchangeably in the literature to refer to application payload. We use the term message to indicate the payload comprising one or more network packets.} completion time (FCT), but requires specialized hardware for embedding scheduling policies within the network fabric.
A centralized time-slot allocation and path assignment algorithm to achieve a \emph{zero-queue} network is proposed in Fastpass \cite{fastpass_sigcomm14}, but centralized schedulers do not scale well, and the matching delay becomes significant compared to the transmission time for short flows.
To address these shortcomings, recent distributed protocols prioritize the transmission of short flows in a manner that circumvents the connection setup delay and its associated round trip time (RTT), whereas receiver-driven flow control is used for sustained connections \cite{phost_conext15,ndp_sigcomm17,homa_sigcomm2018}.
Receiver-driven flow control works well for parallel computations over DCNs because the compute node can rapidly detect and react to local congestion \cite{ousterhout2023itstimereplacetcp}.

The pHost distributed protocol by Gao et al.\ \cite{phost_conext15} was designed to address the shortcomings of Fastpass.
It seeks to minimize FCT in data centers without specialized hardware requirements, assigning high priority to short flow transmissions and lower priority to long flows that require grants from the receiver.
Extending this approach, Montazeri et al.\ introduce Homa \cite{homa_sigcomm2018}, which improves upon pHost by leveraging multiple priority classes (supported by switch priority queues) for short as well as long flows; Homa attains high network utilization through controlled overcommitment.
Underlying these schemes is a tension between the need for fast feedback through control packets and the congestion produced by larger payloads.

The scheduling in many such protocols reduces to matching problems in bipartite graphs, bringing the locality constraints discussed above into direct contact with practical system design.
There is a rich body of research on matching algorithms for network resource allocation \cite{pim_anderson1993high, islip_tnet1999, irrm_el1993, shah2006optimal, mckeown1999achieving, vaze_infocom17_8057223, lowcomplex-tnet2022}.
Parallel iterative matching (PIM) \cite{pim_anderson1993high}, a commonly used algorithm in switch fabrics, proceeds in fixed time slots of matching and transmission phases.
PIM converges to a maximum matching in approximately $\log(n)$ rounds of matching request--responses in an $n$-port switch.
However, such an approach is impractical for DCNs with higher RTTs and large $n$ (number of sender--receiver pairs).
Round-robin schemes such as iSLIP \cite{islip_tnet1999} avoid the randomness of PIM but require knowledge of the system size $n$ and assume that every input has traffic for every output; when the connectivity graph is sparse or irregular, as is typical in DCN workloads, these assumptions break down.
Data center parallel iterative matching (dcPIM) \cite{dcpim_sigcomm2022} adapts PIM to DCNs.
For sparse traffic, dcPIM achieves a high probability of matching within a fixed number of rounds, but the overhead of multiple matching rounds per transmission phase is significant.
In practice, dcPIM seeks to circumvent this overhead through pipelining, wherein the matching phase for the next transmission round happens in parallel with the current transmission phase.

\subsection{Contributions}

This article frames resource allocation as bipartite matching under communication constraints, an abstraction that is portable across application domains including load balancing, network routing, server farms with data locality, and virtual machine migration.
The main contributions are as follows.
\begin{itemize}
\item \emph{Single-round matching algorithms under communication constraints.}
A class of probabilistic bipartite matching algorithms is proposed in which senders disclose their interest in forming connections, receivers share local degree information, and senders then make a single randomized selection.
The entire protocol completes in one round of message exchange.
\item \emph{Random thinning for dense graphs.}
A preemptive thinning step is introduced in which senders report only a random subset of their connections.
For dense connectivity graphs, restricting the choice set in this way increases the expected matching size, even though fewer edges are available to the algorithm.
\item \emph{Degree-biased sampling.}
Senders bias their random selection toward receivers with smaller degrees, producing a larger expected number of matches than uniform selection.
This mechanism is shown to be particularly effective in sparse-traffic regimes.
\item \emph{Robustness to degree distribution.}
The mean matching fraction is shown to be insensitive to the sender out-degree distribution, depending only on the probability that the out-degree is nonzero.
This property makes the algorithms robust under heterogeneous traffic conditions.
\end{itemize}
Analytical performance guarantees are established for random bipartite graph models, including asymptotic results in the Poisson degree regime.
Numerical comparisons against baseline algorithms and established benchmarks complement the theoretical analysis.

\noindent\textbf{Notation.}
We denote the set of 
natural numbers by $\N$, and the non-negative integers by $\Z_+$, and the cardinality of $\sX$ by $\abs{\sX}$. 

\section{System Model}
\label{sec:systemmodel}

We consider a data center network (DCN) with $N$ hosts interconnected through a Clos-like fat-tree topology~\cite{alfares2008scalable}, which provides multiple equal-length paths between any sender--receiver pair.
Each host connects to the network through a Top-of-the-Rack (ToR) switch via an edge link.
We refer to the collection of inter-switch links as the core network.
Since a host can act as both a sender and a receiver simultaneously, the network is modeled as an $N \times N$ system comprising $N$ senders and $N$ receivers.
We assume that the core network has sufficient capacity to support full cross-section bandwidth and that packets can be sprayed over multiple parallel paths for load balancing, as in prior work \cite{homa_sigcomm2018, dcpim_sigcomm2022}.
Congestion is therefore negligible in the core, and the physical network provides full connectivity: every sender can reach every receiver.
The bottleneck lies at the edge links: both the sender's uplink and the receiver's downlink have finite capacity, so each can sustain at most one bulk transfer at a time.

Application traffic in modern DCNs comprises a large number of latency-sensitive \emph{short} messages (control traffic, remote procedure calls) and a smaller number of bandwidth-intensive \emph{long} messages.
Since short and long messages share the same links, long messages can delay short messages through head-of-line blocking in the absence of prioritization.
Following prior proposals \cite{homa_sigcomm2018, dcpim_sigcomm2022}, switch priority classes provide preferential treatment to short messages.
Each ToR switch maintains separate virtual buffers per priority class, and short messages are served before long messages.
We therefore focus on the scheduling of long messages.

Although the physical network is fully connected, not every sender has data destined for every receiver at any given time.
The sender--receiver pairs with outstanding long messages define a \emph{feasible graph}, which is typically a sparse subgraph of the complete bipartite graph.
Because each edge link is dedicated to a single partner during each transmission phase, coordinating sender--receiver assignments reduces to finding a matching in this feasible graph.

\subsection{Mathematical abstraction}
\label{subsec:static}

The system comprises a set of senders $U$ and a set of receivers $V$, both of equal cardinality $N= \abs{U}= \abs{V}$.
An edge $(u,v)$ indicates that sender $u$ has at least one outstanding long message for receiver $v$; the resulting bipartite graph is $G=(U,V,E)$.
We write $u\sim v$, and call $u$ and $v$ neighbors, if $(u,v) \in E$.
The neighborhood of node $x$ is $\nbrh_x \triangleq \set{ y: x \sim y}$, and its degree is $\deg(x) \triangleq \abs{\nbrh_x}$.
We make the following assumption about this graph.
\begin{assum}
\label{assum:DRG}
The bipartite graph $G=(U,V,E)$ is a $D$-out random graph, for a random variable $D$ with a given distribution on $\set{0,1,2,\ldots, N}$.
This random graph is defined as follows. First, nodes $u\in U$ sample degrees $D_u$ that are \iid with the same distribution as $D$.
Then, conditional on $D_u$, node $u$ samples its neighbors in $V$ independently and uniformly at random, and independently of the choices of other nodes.
\end{assum}

An example of a $D$-out random graph is the bipartite Erd\H{o}s-R\'enyi random graph $\mathcal{G}(N,N,p)$, which is obtained by taking $D$ to have the binomial distribution, $\Bin(N,p)$, with parameters $N$ and $p$.
The bipartite graph models only long-message connectivity; the effect of short messages on pair availability can be incorporated by thinning the out-degree distribution, but this extension is not treated explicitly here.

Since each edge link serves at most one partner at a time, any feasible set of communicating sender--receiver pairs corresponds to a matching in $G$.
The main focus of this paper is on finding matchings of large cardinality in order to reduce latency in message transmission.

%

\subsection{Problem formulation}
\label{subsec:problem}

The goal is to find a matching of large cardinality in $G$ using only local information.
Specifically, each sender $u \in U$ knows only its two-hop neighborhood in $G$, i.e., all nodes and edges within graph distance 2.
Based on this knowledge, the sender $u$ selects a single neighbor $v \in \nbrh_u$ to which it sends a grant.
Each receiver $v$ that has received at least one grant accepts exactly one, based on no knowledge of $G$ beyond the received grants.
This grant-and-accept protocol, whose communication stages are detailed in Section~\ref{sec:matching}, establishes a matching.
Finding a matching of maximum cardinality under these constraints is intractable in general, motivating the probabilistic approach developed in the following section.

\section{Bipartite matching algorithms}
\label{sec:matching}
Recall from Section~\ref{sec:systemmodel} that the feasible graph $G=(U,V,E)$ captures all sender--receiver pairs with outstanding long messages, and that a matching in this graph determines the set of concurrent transmissions.
Although this graph is typically sparse, notifying all feasible receivers can still generate significant control traffic in large systems; the problem is especially acute for dense workloads.
To reduce communication and computational cost, we propose that each sender establish links with only a small random subset of its feasible receivers; we call these its \emph{intended receivers} and the resulting subgraph the \emph{intention graph}.

Our goal is to jointly optimize the choice of intended receivers and the matching on the intention graph.
Our proposed algorithms build upon existing matching algorithms for DCN transport protocols~\cite{dcpim_sigcomm2022, pim_anderson1993high}, which we review below.

\subsection{Existing iterative matching algorithms}
The prototypical communication-constrained algorithm, depicted in Fig. \ref{fig:01}, comprises the following stages. 
\begin{itemize}
\item Stage $0$ (NOTIFY): Every sender sends a notification control message to all receivers with which it seeks to connect.
\item Stage $1$ (REQuest): Each receiver sends a request control message (REQ) to all senders from which it received a notification.
\item Stage $2$ (GRANT): Each sender selects one of the received REQs uniformly at random and sends a grant control message (GRANT) to the corresponding receiver.
\item Stage $3$ (ACCEPT): Each receiver selects one of the received GRANTs uniformly at random and transmits an accept control message (ACCEPT) to the corresponding sender. Both the sender and the receiver mark themselves as matched.
\end{itemize}
The above matching algorithm takes approximately two round-trip times (RTTs), assuming instantaneous control message processing and synchronized nodes~\cite{dcpim_sigcomm2022}.
To hide this overhead, the transmission stage of one phase can be pipelined with the matching stage of the next, as the control messages are much smaller than data packets~\cite{dcpim_sigcomm2022}.

\begin{figure}[htb]
\centering
\resizebox{0.9\linewidth}{!}{\begin{tikzpicture}
[node distance=0mm, draw=black, thick, >=stealth',
leaf/.style={rectangle, rounded corners=.5mm, minimum height=6mm, minimum width=6mm, draw=blue!40, fill=blue!20},
labl/.style={rectangle, minimum height=6mm, minimum width=6mm, draw=none, fill=none},
pre/.style={<-},
post/.style={->},
]

\def\hts{1,2,3,4,5,6};
\foreach \y in \hts {
        \node[leaf] (leafNotL\y) at (0,6-\y) {\y};
}

\foreach \y in \hts {
        \node[leaf] (leafNotR\y) at (3,6-\y) {\y};
}

\foreach \y in \hts {
	\node[leaf] (leafReqL\y) at (5,6-\y) {\y};
}
\foreach \y in \hts {
	\node[leaf] (leafReqR\y) at (8,6-\y) {\y};
}

\foreach \y in \hts {
	\node[leaf] (leafGraL\y) at (10,6-\y) {\y};
}
\foreach \y in \hts {
	\node[leaf] (leafGraR\y) at (13,6-\y) {\y};
}

\foreach \y in \hts {
	\node[leaf] (leafAccL\y) at (15,6-\y) {\y};
}
\foreach \y in \hts {
	\node[leaf] (leafAccR\y) at (18,6-\y) {\y};
}

\draw(leafNotR1.west) 
	edge[pre] (leafNotL1.east)
	edge[pre] (leafNotL2.east)
        edge[pre] (leafNotL4.east);
\draw(leafNotR2.west) 
	edge[pre] (leafNotL1.east)
	edge[pre] (leafNotL3.east)
	edge[pre] (leafNotL4.east)
        edge[pre] (leafNotL6.east);
\draw(leafNotR3.west) 
	edge[pre] (leafNotL2.east)
	edge[pre] (leafNotL3.east);
\draw(leafNotR4.west) 
	edge[pre] (leafNotL3.east)
	edge[pre] (leafNotL5.east);
\draw(leafNotR5.west) 
	edge[pre] (leafNotL5.east);
\draw(leafNotR6.west) 
	edge[pre] (leafNotL6.east);
 
\draw(leafReqL1.east) 
	edge[pre] (leafReqR1.west)
	edge[pre] (leafReqR2.west);
\draw(leafReqL2.east) 
	edge[pre] (leafReqR1.west)
	edge[pre] (leafReqR3.west);
\draw(leafReqL3.east) 
	edge[pre] (leafReqR2.west)
	edge[pre] (leafReqR3.west)
	edge[pre] (leafReqR4.west);
\draw(leafReqL4.east) 
	edge[pre] (leafReqR1.west)
	edge[pre] (leafReqR2.west);
\draw(leafReqL5.east) 
	edge[pre] (leafReqR4.west)
	edge[pre] (leafReqR5.west);
\draw(leafReqL6.east) 
	edge[pre] (leafReqR2.west)
	edge[pre] (leafReqR6.west);

\draw(leafGraR1.west) 
	edge[pre] (leafGraL1.east)
	edge[pre] (leafGraL2.east);
\draw(leafGraR2.west) 
	edge[pre] (leafGraL3.east)
	edge[pre] (leafGraL4.east)
	edge[pre] (leafGraL6.east);
\draw(leafGraR4.west) 
	edge[pre] (leafGraL5.east);

\draw(leafAccL2.east) 
	edge[pre, blue!60] (leafAccR1.west);
\draw(leafAccL4.east) 
	edge[pre, blue!60] (leafAccR2.west);
\draw(leafAccL5.east) 
	edge[pre, blue!60] (leafAccR4.west);

\node[labl] at (1.5, 5.3){\bf NOTIFY};
\node[labl] at (6.5, 5.3){\bf REQ};
\node[labl] at (11.5, 5.3){\bf GRANT};
\node[labl] at (16.5, 5.3){\bf ACCEPT};

\node[labl] at (1.5, -1){\Large{\bf Notify stage}};
\node[labl] at (6.5, -1){\Large{\bf Request stage}};
\node[labl] at (11.5, -1){\Large{\bf Grant stage}};
\node[labl] at (16.5, -1){\Large{\bf Accept stage}};
\end{tikzpicture}}
\caption{Existing matching algorithm on a $6 \times 6$ system. Senders notify all feasible receivers; uniform random selection at the GRANT stage causes collisions at receivers~$1$ and~$2$, yielding only three matched pairs.}
\label{fig:01}
\end{figure}
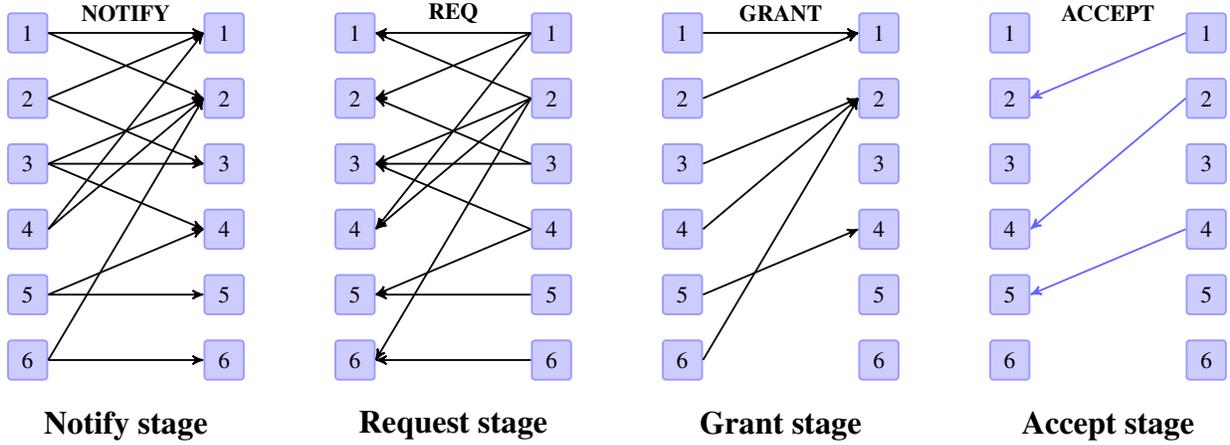

One drawback of the above algorithm is that it can generate excessive control messages because senders notify all feasible receivers.
This can be avoided by having senders only transmit notifications to subsets of their feasible receivers. We incorporate this into our proposed algorithm by sparsifying the communication graph before implementing matching.

\subsection{Degree-biased algorithms}
We introduce a one-parameter family of \emph{degree-biased} (DB) matching algorithms which share much of their architecture with the pre-existing matching framework described above, but incorporate the following innovations.
In the NOTIFY stage, senders only notify a subset of their neighbors, chosen independently without coordination. Thus, this step is decentralized and yields a random subgraph of the initial communication graph.
In the REQ stage, receivers include their degree in the intention graph (the number of senders from which they have received notifications) in their reply. This small overhead enables significant gains.
In the GRANT stage, every sender selects a receiver at random according to a probability distribution that considers their degrees, biasing its choice in favor of low-degree receivers. Empirical results show that this strategy leads to significant performance gains.

Figures \ref{fig:01} and \ref{fig:02} offer an illustrative example with six hosts to highlight the key differences between the degree-biased (Fig. \ref{fig:02}) and pre-existing (Fig.~\ref{fig:01}) algorithms.
In both figures, the nodes on the left correspond to senders and the nodes on the right to receivers.
At the onset of the process, in the pre-existing algorithm depicted in Fig.~\ref{fig:01}, every sender notifies all receivers with which it can feasibly connect.
But in our proposed algorithm, if a sender has an excessive number of notifications to send (node 3 in our example in  Fig. \ref{fig:02}), it subsamples from them and only sends the selected notifications.
This thinning process takes place in every resource allocation round.
The receivers acknowledge notifications by sending requests (REQs) to senders.
In our proposed algorithm, they include the additional information of their degrees (number of sender neighbors) in these REQ messages.
In the GRANT stage, each sender selects one receiver from among its graph neighbors.
This selection is made uniformly at random in existing protocols, but favors low-degree neighbors in our algorithm. We have depicted this by showing senders $\set{5, 6}$ connect to receivers $\set{5, 6}$ under the DB$(\alpha)$ rule in Fig. \ref{fig:02}, but connect to $\set{4, 2}$ in Fig. \ref{fig:01}.

To fully describe our degree-biased algorithm, we need to specify how senders prune their neighborhoods at the NOTIFY stage, and how they choose which receivers to select at the GRANT stage.
At the NOTIFY stage, senders use random thinning to achieve a desired degree distribution.
Each sender $u\in U$ independently samples an out-degree $D_u$ from a specified distribution, then selects $\min(D_u, \abs{\nbrh_u})$ of its feasible receivers uniformly at random.
The resulting intention graph is a $D$-out random bipartite graph as defined in Assumption \ref{assum:DRG}.

\begin{defn}[Thinning]
\label{defn:Thinning}
We refer to the subsampling procedure at the NOTIFY stage as \emph{thinning}: each sender $u \in U$ independently samples an out-degree $D_u$ from a specified distribution and retains $\min(D_u, \abs{\nbrh_u})$ of its feasible receivers, chosen uniformly at random.
\end{defn}

To inform the random selection process at the GRANT stage, we introduce a suitable class of discrete probability laws, as detailed below.
\begin{figure}[htb]
\centering
\resizebox{0.9\linewidth}{!}{\begin{tikzpicture}
[node distance=0mm, draw=black, thick, >=stealth',
leaf/.style={rectangle, rounded corners=.5mm, minimum height=6mm, minimum width=6mm, draw=blue!40, fill=blue!20},
labl/.style={rectangle, minimum height=6mm, minimum width=6mm, draw=none, fill=none},
pre/.style={<-},
post/.style={->},
]

\def\hts{1,2,3,4,5,6};
\foreach \y in \hts {
        \node[leaf] (leafNotL\y) at (0,6-\y) {\y};
}

\foreach \y in \hts {
        \node[leaf] (leafNotR\y) at (3,6-\y) {\y};
}

\foreach \y in \hts {
	\node[leaf] (leafReqL\y) at (5,6-\y) {\y};
}
\foreach \y in \hts {
	\node[leaf] (leafReqR\y) at (8,6-\y) {\y};
}

\foreach \y in \hts {
	\node[leaf] (leafGraL\y) at (10,6-\y) {\y};
}
\foreach \y in \hts {
	\node[leaf] (leafGraR\y) at (13,6-\y) {\y};
}

\foreach \y in \hts {
	\node[leaf] (leafAccL\y) at (15,6-\y) {\y};
}
\foreach \y in \hts {
	\node[leaf] (leafAccR\y) at (18,6-\y) {\y};
}

\draw(leafNotR1.west) 
	edge[pre] (leafNotL1.east)
	edge[pre] (leafNotL2.east)
        edge[pre] (leafNotL4.east);
\draw(leafNotR2.west) 
	edge[pre] (leafNotL1.east)
	edge[dashed, <-, red!60] (leafNotL3.east)
	edge[pre] (leafNotL4.east)
        edge[pre] (leafNotL6.east);
\draw(leafNotR3.west) 
	edge[pre] (leafNotL2.east)
	edge[pre] (leafNotL3.east);
\draw(leafNotR4.west) 
	edge[pre] (leafNotL3.east)
	edge[pre] (leafNotL5.east);
\draw(leafNotR5.west) 
	edge[pre] (leafNotL5.east);
\draw(leafNotR6.west) 
	edge[pre] (leafNotL6.east);
 
\draw(leafReqL1.east) 
	edge[pre] (leafReqR1.west)
	edge[pre] (leafReqR2.west);
\draw(leafReqL2.east) 
	edge[pre] (leafReqR1.west)
	edge[pre] (leafReqR3.west);
\draw(leafReqL3.east) 
	edge[pre] (leafReqR3.west)
	edge[pre] (leafReqR4.west);
\draw(leafReqL4.east) 
	edge[pre] (leafReqR1.west)
	edge[pre] (leafReqR2.west);
\draw(leafReqL5.east) 
	edge[pre] (leafReqR4.west)
	edge[pre] (leafReqR5.west);
\draw(leafReqL6.east) 
	edge[pre] (leafReqR2.west)
	edge[pre] (leafReqR6.west);

\draw(leafGraR1.west) 
	edge[pre] (leafGraL1.east);
\draw(leafGraR2.west)
        edge[pre] (leafGraL4.east);
\draw(leafGraR3.west)
        edge[pre] (leafGraL2.east);
\draw(leafGraR4.west)
	edge[pre] (leafGraL3.east);
\draw(leafGraR5.west)
        edge[pre] (leafGraL5.east);
\draw(leafGraR6.west)
        edge[pre] (leafGraL6.east);

\draw(leafAccL1.east) 
	edge[pre, blue!60] (leafAccR1.west);
\draw(leafAccL2.east)
        edge[pre, blue!60] (leafAccR3.west);
\draw(leafAccL3.east)
        edge[pre, blue!60] (leafAccR4.west);
\draw(leafAccL4.east)
        edge[pre, blue!60] (leafAccR2.west);
\draw(leafAccL5.east) 
	edge[pre, blue!60] (leafAccR5.west);
\draw(leafAccL6.east)
        edge[pre, blue!60] (leafAccR6.west);

\node[labl] at (1.5, 5.3){\bf NOTIFY};
\node[labl] at (6.5, 5.3){\bf REQ};
\node[labl] at (11.5, 5.3){\bf GRANT};
\node[labl] at (16.5, 5.3){\bf ACCEPT};

\node[labl] at (1.5, -1){\Large{\bf Notify stage}};
\node[labl] at (6.5, -1){\Large{\bf Request stage}};
\node[labl] at (11.5, -1){\Large{\bf Grant stage}};
\node[labl] at (16.5, -1){\Large{\bf Accept stage}};
\end{tikzpicture}}
\caption{Degree-biased matching on the same $6 \times 6$ system. Sender~$3$ thins one edge (dashed) at the NOTIFY stage; degree-biased selection at the GRANT stage spreads grants across receivers, yielding six matched pairs.}
\label{fig:02}
\end{figure}
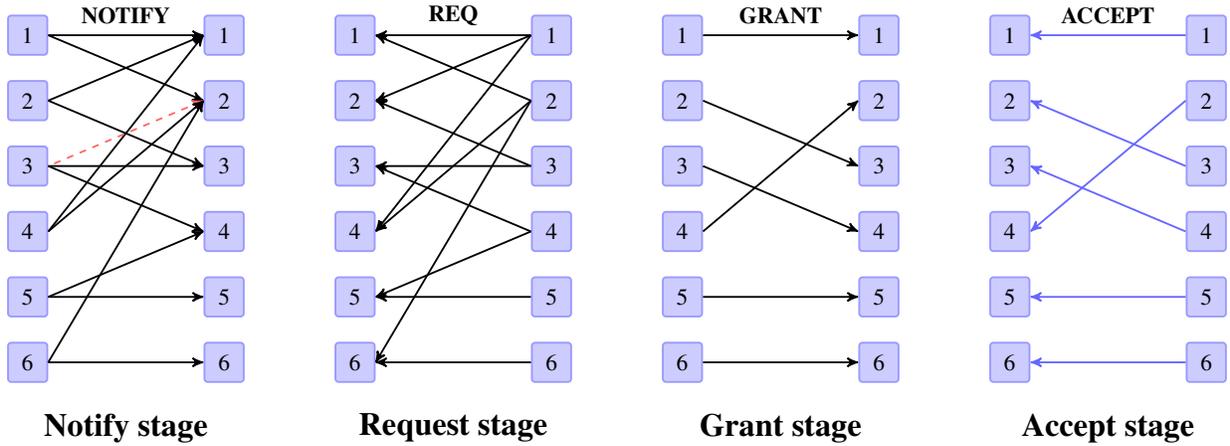

\begin{defn}[DB$(\alpha)$ selection]
\label{defn:alphaprob}
Consider a bipartite graph $G = (U,V,E)$ with senders $U$, receivers $V$, edges $E \subseteq U\times V$,
and a parameter $\alpha \in \R \cup \set{-\infty}$. The DB$(\alpha)$ algorithm defines a probability mass function (PMF) on $\nbrh_u$ for each $u\in U$ with non-zero degree:
\begin{equation}
\label{eqn:DBalpha}
\gamma_{uv}(\alpha) \triangleq \SetIn{u \sim v}\frac{\deg(v)^\alpha}{\sum_{w \in \nbrh_u}\deg(w)^\alpha}.
\end{equation}
\end{defn}
We note that the computation of this PMF at node $u\in U$ relies on knowledge of $\nbrh_u$, the neighbors of sender $u$, and on $\set{\deg(v): v \in \nbrh_u}$, the size of the neighborhood of each receiver $v$ in the neighborhood of sender $u$.
This local information is gathered in the NOTIFY and REQ stages, as described above.

A natural question is to find the exponent $\alpha^\ast$ in the one-parameter family DB$(\alpha)$ that maximizes the mean matching size, averaged over the ensemble of $D$-out random graphs.
Formally, we can pose this problem as follows.
\begin{probl}
\label{prob:OptExp}
Let $D$ be a random variable with a specified distribution. Let $L_N(\alpha)$ be the cardinality of the matching in a bipartite random $D$-out graph on $N$ nodes induced by the
DB$(\alpha)$ algorithm (or equivalently, by the PMFs $(\gamma_{u}(\alpha): u \in U, \deg(u) \neq 0)$).
Find the exponent $\alpha^\ast \in \R$ that maximizes the mean fraction of matched receivers
as the number of receivers $N$ tends to infinity, i.e.
\EQN{
\label{eqn:prob}
\alpha^\ast \triangleq \operatorname*{arg\,max}\set{\liminf_{N\to\infty}\frac1N\E L_N(\alpha): \alpha \in \R}.
}
\end{probl}
A key step in solving this problem is evaluating the limiting mean matching fraction, $\lim_{N\to\infty}\frac1N\E L_N(\alpha)$, for all $\alpha \in \R$; as part of this effort, we show the existence of this limit in certain cases.

Observe that the DB$(\alpha)$ algorithm favors receivers with higher degree when $\alpha >0$ and with lower degree when $\alpha < 0$. As selecting receivers with higher degrees leads to poor matching performance, we focus solely on the case $\alpha \le 0$.
In the next section, we compute the limiting mean matching fraction for the extreme cases, $\alpha=0$ and $\alpha=-\infty$, and augment this with numerical evaluations for intermediate values of $\alpha$.

\section{Performance analysis of DB(\texorpdfstring{$\alpha$}{alpha})}
\label{sec:Analysis}

In this section, we derive expressions for the mean matching size under the DB$(\alpha)$ rule for selection of receivers, obtaining closed-form expressions when $\alpha$ is equal to zero or negative infinity.
Note that $\alpha=0$ corresponds to selecting uniformly at random among receivers neighboring the sender, while $\alpha=-\infty$ corresponds to restricting the selection to neighboring receivers with the smallest degree.
For these two cases, the DB$(\alpha)$ rule is equivalent to selecting a receiver uniformly at random from the admissible set.

\begin{assum}[Intention graph] \label{assum:DRG-IG}
The intention graph $G=(U,V,E)$ is a $D$-out random bipartite graph, i.e., the sender degrees $(D_u, u\in U)$ are \iid copies of $\min\set{D, N}$, where $D$ is a specified random variable on $\Z_+$ and $N=\abs{U}=\abs{V}$.
Moreover, conditional on $D_u$, the neighbors of $u$ are chosen independently and uniformly at random from $V$, independent of the choices of other senders.
\end{assum}

\begin{rem}
\label{rem:EdgeProb}
By Assumption~\ref{assum:DRG-IG}, the conditional expectation $\E[\SetIn{v \in \nbrh_u}\mid D_u]$ equals $D_u/N$.
Hence, by the tower property of nested expectation, 
\EQN{ \label{eq:edge_prob}
\E\SetIn{u \sim v} = \E[\E[\SetIn{v \in \nbrh_u}\mid D_u]] = \frac{\E[D_u]}{N} = \frac{\meandeg}{N}. 
}
Moreover, for fixed $v\in V$, the random variables $\SetIn{u \sim v}$ are \iid across $u\in U$, since distinct senders choose neighbors independently.
Hence, $\deg(v) = \sum_{u \in U}\SetIn{u \sim v}$ is a $\Bin(N, {\meandeg}/{N})$ random variable for each $v\in V$. 
\end{rem}

The following lemma will be useful in analyzing the performance of our matching algorithms.

\begin{lem}
\label{lem:cond_degree}
Let $u\in U$ and $v\in V$.
Let $D^\ast_u$ be a random variable with the conditional distribution of the degree of $u$ given that $u \sim v$. 
Then, $D^\ast_u$ has the size-biased distribution, i.e., for each $k \in \set{1, \dots, N}$
\EQ{
P\set{D^\ast_u=k} = P(\set{D_u=k}\mid\set{u \sim v}) = \frac{kP\set{D=k}}{\meandeg}. 
}
\end{lem}
\begin{IEEEproof}
By the conditional probability formula,
\begin{align*}
P\set{D^\ast_u=k}
&= P\set{D_u=k} \frac{\E[\SetIn{u \sim v}\mid\set{D_u=k}]}{\E\SetIn{u \sim v}}.
\end{align*} 
The claim follows from Remark \ref{rem:EdgeProb}.
\end{IEEEproof}

\begin{defn}
\label{defn:ER}
Under the DB$(\alpha)$ selection rule, let $\xi_{uv}(\alpha)$ indicate the event that sender $u$ selects receiver $v$ for a grant.
Let $\psi_v(\alpha)$ indicate the event that receiver $v$ is not matched to any sender.
Finally, let $L_N(\alpha)$ denote the size of the matching obtained.
\end{defn}

\begin {rem}
\label{rem:MatchSize}
As each sender issues a grant to a single receiver, a receiver is matched if and only if it receives at least one grant. 
Hence, 
\meq{2}{
&\psi_v(\alpha) = \prod_{u \in \nbrh_v}(1-\xi_{uv}(\alpha)),&&L_N(\alpha)= \sum_{v \in V}(1-\psi_v(\alpha)).
}
\end{rem}

\subsection{Uniform selection: $\alpha = 0$}

\begin{thm}
\label{thm:RandRxUnif}
Suppose $G=(U,V,E)$ satisfies Assumption~\ref{assum:DRG-IG}.
Then, the mean matching fraction under the DB$(0)$ algorithm is given by
\EQN{
\label{eqn:MeanMatchUnif}
\frac{1}{N} \E L_N(0) = 1 - \left(1 - \frac{1-P\set{D=0}}{N} \right)^N.
}
\end{thm}

\begin{rem}
The mean matching fraction under DB$(0)$ is insensitive to the out-degree distribution and depends only on the probability that the out-degree is non-zero.
If thinning does not change this probability, then it has no effect on the mean matching fraction under this rule.
\end{rem}

\begin{IEEEproof}
Fix $v\in V$ and consider $u\in U$. If $D_u=0$, then $u$ cannot issue a grant to $v$.
Otherwise, as it only issues a grant to a single receiver, it picks $v$ with probability $1/N$.
(While $u$ first picks a neighborhood and then picks a node within this neighborhood, by symmetry, each node in $V$ is equally likely to be chosen at the end of this two-step process.)
Hence, the overall probability that $u$ issues a grant to $v$ is 
\EQN{ 
\label{eq:meanxi_alpha0}
\E[ \xi_{uv}(0)] = \frac{1-P\set{D_u=0}}{N}.
}
Another way to see this is to note that the event $u \sim v$ has probability $\meandeg/N$ and, conditional on this event, $u$ chooses $v$ for a grant with probability $1/D^\ast_u$, where $D^\ast_u$ has the size-biased distribution from Lemma~\ref{lem:cond_degree}. 
Hence, using Lemma~\ref{lem:cond_degree}, 
\begin{align*}
\E[\xi_{uv}(0)] &= \E[\xi_{uv}(0) \SetIn{v\in \nbrh_u}] = \E\SetIn{u \sim v} \E[\xi_{uv}(0)\mid v\in \nbrh_u] \\
&= \frac{\meandeg}{N}\sum_{k=1}^{N} \frac{1}{k} P\set{D^\ast=k}
= \frac{1}{N} \sum_{k=1}^{N} P\set{D=k},
\end{align*}
which coincides with the expression in~\eqref{eq:meanxi_alpha0}. 
Moreover, by the independence of sender neighborhoods in Assumption~\ref{assum:DRG-IG}, $(\xi_{uv}(0),u\in U)$ are the indicators of mutually independent events. 
Hence,
\EQ{
\E[\psi_v(0)] = \prod_{u\in U} (1-\E[\xi_{uv}(0)])
= \left( 1- \frac{1-P\set{D=0}}{N} \right)^N,
}
and \eqref{eqn:MeanMatchUnif} follows from Remark~\ref{rem:MatchSize}. 
\end{IEEEproof}

\begin{cor}
\label{cor:RandRxUnif}
Two easy special cases for sender degree distributions for a fixed positive mean $\meandeg$ are as follows. 
\begin{compactenum}[(a)]
\item For deterministic sender degree $D = \meandeg \in \mathbb{N}$, the mean matching fraction is 
\EQ{
\frac1N\E L_N(0) = 1 - \left(1 - \frac{1}{N} \right)^N.
} 
The limiting mean matching fraction is 
\EQ{
\lim_{N\to\infty}\frac1N\E L_N(0) = 1- \frac1e.
}
\item For binomial sender degree distribution $\Bin(N, \meandeg/N)$, the mean matching fraction is
\EQ{
\frac1N\E L_N(0) = 1 - \left(1 - \frac{1-(1-\frac{\meandeg}{N} )^N}{N}\right)^N. 
}
The limiting mean matching fraction is 
\EQ{
\lim_{N\to\infty}\frac1N\E L_N(0) = 1-e^{-1+e^{-\meandeg}}.
}
\end{compactenum}
\end{cor}

\subsection{Asymptotic independence}
We can no longer carry out exact calculations when $\alpha \neq 0$, because grant decisions of distinct senders may no longer be independent.
Nevertheless, they are independent for senders with disjoint neighborhoods. This fact lets us approximate matching probabilities under DB$(\alpha)$ with simpler expressions calculated under an independence assumption.
In the following, we focus on the limiting regime in which $N$ tends to infinity, while the sender degree distribution remains fixed.
Let $\meandeg$ and $\E[D^2]$ denote its first and second moments.

\begin{lem} \label{lem:ProbDisconn1}
Fix $\epsilon>0$ and $v\in V$, and suppose that 
\EQ{
\deg(v) < \min \set{ \sqrt{\frac{N (\meandeg)^2}{(\E D^2)^2}}, \sqrt[3]{\frac{\epsilon N (\meandeg)^2}{(\E D^2)^2}} }.
}
Then, $\E[\psi_v(\alpha)\mid \nbrh_v]$, the conditional probability that node $v$ is unmatched given its neighborhood, is upper bounded by 
\begin{align*}
& e^{\epsilon} (1-\E[\xi_{uv}(\alpha)\mid \nbrh_v, u \in \nbrh_v])^{\deg(v)} + \frac{1}{2}\sqrt[3]{\frac{\epsilon^2 (\E D^2)^2}{N (\meandeg)^2}}.
\end{align*}
\end{lem}
\begin{IEEEproof}
When $\deg(v) \in \set{0,1}$, the statement of the lemma is immediate.
For each receiver $v$ with $\deg(v) \ge 2$, we define the good event
\begin{equation}
\label{eqn:GE}
G_v \triangleq \Big\{\bigcup_{\set{k,\ell} \subseteq \nbrh_v} (\nbrh_k \cap \nbrh_\ell\setminus\set{v}) = \emptyset\Big\}.
\end{equation}
Recall from Remark \ref{rem:MatchSize} that $\psi_v(\alpha)$, the indicator of the event that receiver $v$ remains unmatched, equals $\prod_{u\in\nbrh_v} \big( 1-\xi_{uv}(\alpha) \big)$. 
As this product is bounded above by $1$, we obtain 
\begin{equation*}
\begin{aligned}
&\E[\psi_v(\alpha)\mid\nbrh_v]= \E\left[ \prod_{u\in\nbrh_v}(1-\xi_{uv}(\alpha)) \Bigm| \nbrh_v \right] \\
&\le \E \left(\prod_{u\in\nbrh_v}(1-\xi_{uv}(\alpha))\Ind{G_v} \Bigm| \nbrh_v \right) + \E [\Ind{G_v^c} \mid \nbrh_v] \\
&\le \E \left(\prod_{u\in\nbrh_v}(1-\xi_{uv}(\alpha)) \Bigm| G_v, \nbrh_v \right) + P(G_v^c \mid \nbrh_v).
\end{aligned}
\end{equation*}
Note that, conditioned on the sender set $\nbrh_v$ and the good event $G_v$, the random variables $(\xi_{uv}(\alpha), u\in\nbrh_v)$ are mutually independent. 
Using this independence, the expression above, and the fact that $1-\xi_{uv}(\alpha)\ge 0$, we obtain
\begin{equation} \label{eq:unmatched_prob}
\begin{split}
&\E[\psi_v(\alpha)\mid\nbrh_v]
\le \prod_{u\in \nbrh_v} \E[1-\xi_{uv}(\alpha) \mid G_v, \nbrh_v] + P(G_v^c \mid \nbrh_v) \\
&= \prod_{u\in \nbrh_v} \frac{\E [(1-\xi_{uv}(\alpha)) \Ind{G_v}\mid \nbrh_v]}{P(G_v\mid \nbrh_v)} + P(G_v^c \mid \nbrh_v) \\
&\le \prod_{u\in \nbrh_v} \frac{\E [(1-\xi_{uv}(\alpha))\mid \nbrh_v]}{P(G_v\mid \nbrh_v)} + P(G_v^c \mid \nbrh_v).
\end{split}
\end{equation}

Next, we derive an upper bound on $P(G_v^c|\nbrh_v)$.
Fix distinct senders $k,\ell \in \nbrh_v$ and notice that the receiver sets $\nbrh_k$ and $\nbrh_\ell$ are independent. 
Moreover, conditional on $k,\ell \in \nbrh_v$, $\deg(k)$ and $\deg(\ell)$ have the size-biased distribution $D^\ast$. Hence,
\begin{equation*}
\begin{aligned}
&\E\abs{\nbrh_k \cap \nbrh_{\ell}\setminus\set{v}}
= \sum_{w\in V\setminus\set{v}} P\set{w \in \nbrh_k} P\set{w \in \nbrh_\ell} \\
&= (N-1)\left(\frac{\E{D^\ast}-1}{N-1}\right)^2 \le \frac{(\E D^\ast)^2}{N}
= \frac{(\E D^2)^2}{N(\meandeg)^2},
\end{aligned}
\end{equation*}
where we have used Lemma~\ref{lem:cond_degree} to obtain the last equality.
Consequently, by the union bound, we get
\begin{equation*}
\begin{aligned}
&\E \Big[\abs{\cup_{\set{k,\ell} \subseteq \nbrh_v} \nbrh_k \cap \nbrh_\ell\setminus\set{v}} \Bigm| \nbrh_v \Big] \\
&\le \E \Big[ \sum_{\set{k,\ell} \subseteq \nbrh_v} \abs{\nbrh_k \cap \nbrh_\ell\setminus\set{v}} \Bigm| \nbrh_v \Big] \\
&\le \binom{\abs{\nbrh_v}}{2} \frac{(\E D^2)^2}{N(\meandeg)^2}.
\end{aligned}
\end{equation*}
We observe that, from the definition of $G_v$ in \eqref{eqn:GE}, we have $G_v^c = \set{\abs{\bigcup_{\set{k,\ell} \subseteq \nbrh_v} (\nbrh_k \cap \nbrh_\ell\setminus\set{v})} \ge 1}$.
Therefore, by Markov's inequality,
\begin{equation}
\label{eq:good_event_prob}
\begin{split}
P(G_v^c \mid \nbrh_v)
&\le \E \Bigl[\abs{\cup_{\set{k,\ell} \subseteq \nbrh_v} \nbrh_k \cap \nbrh_\ell\setminus\set{v}} \Bigm| \nbrh_v \Bigr] \\
&\le \frac{\deg(v)^2}{2N} \frac{(\E D^2)^2}{(\meandeg)^2},
\end{split}
\end{equation}
where we used $\abs{\nbrh_v}=\deg(v)$ in the bound for the binomial coefficient.
Substituting \eqref{eq:good_event_prob} into \eqref{eq:unmatched_prob}, and noting that $\E[\xi_{uv}(\alpha)\mid\nbrh_v]$ is the same for all senders $u \in \nbrh_v$ by symmetry, we get that
\begin{equation}
\label{eq:unmatched_prob_bd1}
\begin{split}
\E[\psi_v(\alpha)\mid \nbrh_v] &\le \frac{\deg(v)^2(\E D^2)^2}{2N(\meandeg)^2}\\
&+\Bigl(\frac{1-\E[\xi_{uv}(\alpha)\mid \nbrh_v, u\in \nbrh_v]}{1-\frac{\deg(v)^2 (\E D^2)^2}{2N (\meandeg)^2}}\Bigr)^{\deg(v)}.
\end{split}
\end{equation}
By assumption, we have $\deg(v)^2(\E D^2)^2/(2N(\meandeg)^2)<1/2$.
Combining this with the fact that $1-x\ge e^{-2x}$ on $[0,1/2]$, we obtain
\EQ{
\Bigl( 1-\frac{\deg(v)^2(\E D^2)^2}{2N(\meandeg)^2} \Bigr)^{-\deg(v)} \le \exp\Bigl(\frac{\deg(v)^3(\E D^2)^2}{N(\meandeg)^2}\Bigr).
}
Also by assumption, $\deg(v)< \sqrt[3]{\frac{\epsilon N (\meandeg)^2}{(\E D^2)^2}}$ and, hence,
\meq{2}{
&\frac{\deg(v)^3 (\E D^2)^2}{N (\meandeg)^2} < \epsilon,&&\frac{(\E D^2)^2\deg(v)^2}{2N (\meandeg)^2} < \frac12\sqrt[3]{\frac{\epsilon^2(\E D^2)^2}{N(\meandeg)^2}}.
}
Substituting these results in ~\eqref{eq:unmatched_prob_bd1}, we obtain the claim of the lemma.
\end{IEEEproof}

\begin{rem}
From Remark \ref{rem:EdgeProb}, $\deg(v)$ is a $\Bin(N, {\meandeg}/{N})$ random variable.
Therefore, the assumption in Lemma~\ref{lem:ProbDisconn1} that $\deg(v)$ is bounded by the minimum of the two stated terms holds with high probability as $N$ tends to infinity.
The additive error term in the lemma tends to zero as $N$ tends to infinity, and the multiplicative factor $e^{\epsilon}$ can be made arbitrarily close to $1$ by choosing $\epsilon$ small.
Hence, the lemma says that we can approximate the conditional probability of receiver $v$ being unmatched, given its neighborhood $\nbrh_v$, by the product of probabilities that it is unmatched to each of its neighbors.
Thus, the selection indicators $(\xi_{uv}(\alpha): u \in \nbrh_v)$ become asymptotically independent, with error bounded by the lemma.
\end{rem}

\subsection{Greedy selection: $\alpha=-\infty$}
In this section, we derive a lower bound on the mean matching size when employing the degree-biased algorithm with exponent $\alpha=-\infty$.
In this case, a sender $u$ restricts its grant to receivers with the smallest degree among its neighbors.
If there is more than one such receiver, one is chosen uniformly at random.
We first introduce some notation.

\begin{defn}
\label{defn:PalmBinom}
For a Binomial random variable with parameters $(N,p)$, we denote its probability mass function by $q(N,p)$, its distribution function by $Q(N,p)$, and its complementary distribution function by $\bar{Q}(N,p)$.
That is, for $0\le k\le N$,
\begin{align*}
q_k(N,p) &\triangleq \binom{N}{k}p^k(1-p)^{N-k}, \\
Q_k(N,p) &\triangleq \sum_{\ell=0}^k q_\ell(N,p), \\
\bar{Q}_k(N,p) &\triangleq 1-Q_k(N,p).
\end{align*}
\end{defn}

\begin{lem}
\label{lem:BinomReciprocal}
Suppose $X$ has a Binomial distribution with parameters $N$ and $p>0$.
Then, for all $\theta \in \R$,
\EQ{
\E \left[\frac{\theta^{X+1}}{X+1} \right] = \frac{(1-p+p\theta)^{N+1}-(1-p)^{N+1}}{(N+1)p}.
}
\end{lem}
\begin{IEEEproof}
Since the distribution of $X$ is Binomial with parameters $(N,p)$, we can write
\EQ{
\E \left[ \frac{\theta^{X+1}}{X+1} \right] = \sum_{k=0}^N\frac{1}{k+1}\binom{N}{k} \theta^{k+1} p^k(1-p)^{N-k}.
}
Using the fact that $\frac{N+1}{k+1}\binom{N}{k} = \binom{N+1}{k+1}$ and a change of variables in the summation index, we can rewrite the above equation as
\EQ{
\E \left[\frac{\theta^{X+1}}{X+1} \right] = \frac{1}{(N+1)p}\sum_{k=1}^{N+1}\binom{N+1}{k}(\theta p)^{k}(1-p)^{N+1-k}.
}
The expression is then obtained from the binomial expansion of $(1-p+p\theta)^{N+1}$.
\end{IEEEproof}

\begin{lem}
\label{lem:ProbConnGreedy}
Fix sender $u\in U$ with out-degree $D_u$.
Suppose that the residual degrees of its neighbors, namely the number of their neighbors excluding $u$, are mutually independent $\Bin(N-1,\meandeg/N)$ random variables. 
Under the DB($-\infty$) rule, we have
\begin{equation*}
\begin{split}
&\E \big[\xi_{uv}(-\infty)\mid \nbrh_u, \deg(v)=s \big] \\
&= \SetIn{u \in \nbrh_v} \frac{\bar{Q}_{s-2}(N-1,\frac{\meandeg}{N} )^{D_u}-(\bar{Q}_{s-1}(N-1,\frac{\meandeg}{N}))^{D_u}}{D_uq_{s-1}(N-1,\frac{\meandeg}{N})}
\end{split}
\end{equation*}
where the expression is understood to be zero when $s-1$ lies outside the support of the residual degree distribution.
\end{lem}

\begin{rem}
Each sender $w\in U$ with degree $D_w$, $w\neq u$, includes $v\in V$ in its neighborhood with probability $D_w/N$ conditional on its degree and, therefore, with unconditional probability $\E D/N$. 
Moreover, choices by distinct nodes in $U$ are mutually independent. 
Hence, the neighbors of $u$ have the residual degree distribution claimed in the lemma, and the real content of the assumption is their independence. Now, independence does not hold in general for finite $N$, unless the out-degree distribution is Binomial, in which case we have a bipartite Erd\H{o}s-R\'enyi random graph. 
Nevertheless, it is reasonable to expect asymptotic independence under mild conditions on the distribution of $D$, such as the finiteness of some number of moments. 
For example, it is known that the configuration model exhibits local weak (or Benjamini-Schramm) convergence to a Galton-Watson tree, which implies independence of node degrees; see, e.g.,~\cite[Theorem 4.1]{RvdH2}. 
We were unable to find similar results for bipartite graphs, and the pursuit of one is well outside the scope of this paper, but Lemma~\ref{lem:ProbDisconn1} suggests that our assumption is plausible.
\end{rem}

\begin{IEEEproof}
Fix receiver $v\in V$ with $\deg(v)=s$. 
The lemma holds trivially if $s=0$. 
Therefore, we assume $s\ge 1$ without loss of generality.
Under DB$(-\infty)$, sender $u$ can issue a grant to receiver $v$ only if $v$ has the smallest degree among the neighbors of $u$. 
Let $\chi_{uv}$ be the indicator of this event, i.e.,
\EQ{
\chi_{uv} = \SetIn{u \in \nbrh_v} \prod_{w \in \nbrh_u} \SetIn{\deg(w) \ge \deg(v)}. 
}
Conditional on this event, $v$ is chosen for a grant with probability $1/K_u(s)$, where $K_u(s)$ is the number of neighbors of $u$ with degree $s$:
\begin{equation*}
K_u(s) \triangleq \abs{\set{w\in \nbrh_u: \deg(w)=s}}.
\end{equation*} 

For all neighbors $w$ of $u$ to have $\deg(w) \ge s$, they must all have residual degree at least $s-1$. 
By the assumption in the lemma that the residual degrees of the neighbors of $u$ are independent Binomial random variables, we obtain 
\begin{equation} \label{eq:prob_min_deg_s}
\begin{split}
&\E[\chi_{uv}\mid \nbrh_u, \deg(v)=s] \\
&= \SetIn{u \in \nbrh_v} \bar{Q}_{s-2}\Bigl( N-1,\frac{\meandeg}{N} \Bigr)^{D_u-1}.
\end{split}
\end{equation}
This assumption further implies that for any sender $u \in \nbrh_v$, conditional on its degree $D_{u}$, the event $\chi_{uv}=1$, and $\deg(v) = s$, the random variable $K_u(s)-1$ has a Binomial distribution with parameters
\begin{equation*}
\left( D_u-1, \frac{q_{s-1}(N-1,\meandeg/N)}{\bar{Q}_{s-2}(N-1,\meandeg/N)} \right). 
\end{equation*}
Hence, by Lemma~\ref{lem:BinomReciprocal},
\begin{equation} \label{eq:condprob_grant}
\begin{split}
&\E \Bigl[ \frac{\chi_{uv}}{K_u(s)} \Bigm| \chi_{uv}, \nbrh_u, \deg(v)=s \Bigr] \\ 
&= \frac{ \chi_{uv} \left( 1-\bigl( 1-\frac{q_{s-1}(N-1,\meandeg/N)}{\bar{Q}_{s-2}(N-1,\meandeg/N)} \bigr)^{D_u} \right)}{D_u q_{s-1}(N-1,\meandeg/N)/\bar{Q}_{s-2}(N-1,\meandeg/N)} \\ 
&= \chi_{uv} \frac{ \bar{Q}_{s-2}\bigl( N-1,\frac{\meandeg}{N} \bigr) \Bigl( 1-\bigl( \frac{\bar{Q}_{s-1}(N-1,\frac{\meandeg}{N})}{\bar{Q}_{s-2}(N-1,\frac{\meandeg}{N})} \bigr)^{D_u} \Bigr)}{D_u q_{s-1}(N-1,\frac{\meandeg}{N})}.
\end{split}
\end{equation}

Since $u$ chooses $v$ for a grant with probability $1/K_u(s)$ if $v$ is a minimum-degree neighbor of $u$, and with probability zero otherwise, we see that 
\begin{equation} \label{eq:xi_cond_chi}
\begin{split}
&\E[\xi_{uv}(-\infty) \mid \chi_{uv}, \nbrh_u, \deg(v)=s] \\ 
&= \E \Bigl[ \frac{\chi_{uv}}{K_u(s)} \Bigm| \chi_{uv}, \nbrh_u, \deg(v)=s \Bigr].
\end{split}
\end{equation}
Taking expectations on both sides of \eqref{eq:xi_cond_chi} conditional on $\nbrh_u, \deg(v)=s$, using \eqref{eq:prob_min_deg_s} and \eqref{eq:condprob_grant}, we obtain from the tower property of conditional expectation that 
\begin{equation*}
\begin{aligned}
& \E [ \xi_{uv}(-\infty) \mid \nbrh_u, \deg(v)=s] \\
&= \E\Bigg[\E\Bigg[\frac{\chi_{uv}}{K_u(s)}\Bigm|\chi_{uv},\nbrh_u, \deg(v)=s\Bigg]\Bigm| \nbrh_u, \deg(v)=s\Bigg]\\ 
&= \SetIn{u \in \nbrh_v} \frac{\bar{Q}_{s-2}(N-1,\frac{\meandeg}{N})^{D_u} - \bar{Q}_{s-1}(N-1,\frac{\meandeg}{N})^{D_u}}{D_u q_{s-1}(N-1,\frac{\meandeg}{N})},
\end{aligned}
\end{equation*}
as claimed.
\end{IEEEproof}

We now consider a sparse graph asymptotic regime in which $N$, the number of senders (and receivers), tends to infinity, while the mean degree, $\meandeg$, remains constant.
In this regime, the receiver degree distribution tends to a Poisson distribution.
This motivates us to define the following. 

\begin{defn}
Denote by $\pi$ and $\Pi$ the probability mass function and distribution function respectively of a Poisson random variable with parameter $\meandeg$.
The dependence of $\pi$ and $\Pi$ on $\meandeg$ is suppressed in the notation.
\end{defn}

Using Lemmas~\ref{lem:ProbDisconn1} and~\ref{lem:ProbConnGreedy}, we derive the conditional probability that a receiver $v$ is unmatched in the limiting regime under consideration.
The out-degree distribution enters the result through its generating function. 
\begin{defn}
Let $G_D$ denote the generating function of a random variable $D$ with the out-degree distribution of senders, i.e., $G_D(z)=\E[z^{D_u}]$ for $u\in U$. 
\end{defn}

\begin{thm}
\label{thm:RandRxGreedy}
Consider a limiting regime in which $N$ tends to infinity while the mean sender degree $\meandeg$ is a constant that does not depend on $N$.
Then, 
\begin{equation*}
\liminf_{N\to \infty}\frac1N\E L_N(-\infty) \ge 1 - \sum_{s=0}^{\infty} e^{-\meandeg}\frac{(\meandeg)^s}{s!}(1-f(s))^s,
\end{equation*}
where $f:\Z_+ \to[0,1]$ is defined by $f(0)\triangleq 0$ and, for $s\ge 1$,
\begin{equation}
\label{eqn:ConProbGreedy4}
f(s) \triangleq
\frac{G_D(\bar\Pi_{s-2})- G_D(\bar\Pi_{s-1})}{\meandeg \pi_{s-1}}.
\end{equation}
\end{thm}
\begin{IEEEproof}
We recall that conditional on $u\in \nbrh_v$, the random variable $D_u$ has the size-biased distribution given in Lemma~\ref{lem:cond_degree}.
Hence, we see from the definition of $G_D$ that 
\EQN{
\label{eqn:MeanRatio}
\begin{split}
\E \Bigl[ \frac{z^{D_u}}{D_u} \Bigm| u\in \nbrh_v \Bigr] &= \sum_{k=1}^N \frac{kP\set{D_u=k}}{\E D}\frac{z^k}{k} \\
&= \frac{G_D(z)-\E\SetIn{D=0}}{\E D}.
\end{split}
}
Conditioning on $u \in \nbrh_v$ and then taking the expectation over $\nbrh_u$ on both sides of the result in Lemma~\ref{lem:ProbConnGreedy}, and using the tower property of conditional expectation and \eqref{eqn:MeanRatio}, we obtain that 
\begin{equation} \label{eqn:PreLimit}
\begin{aligned}
&\E \big[\xi_{uv}(-\infty)\mid u \in \nbrh_v, \deg(v)=s \big] \\
&= \frac{G_D(\bar{Q}_{s-2}(N-1,\frac{\meandeg}{N})) - G_D(\bar{Q}_{s-1}(N-1,\frac{\meandeg}{N}))}{\E D q_{s-1}(N-1,\frac{\meandeg}{N})} .
\end{aligned}
\end{equation}
In the limiting regime considered, the Binomial distribution with parameters $N-1$ and $\meandeg/N$ converges to a Poisson distribution with parameter $\meandeg$, i.e., for each $k \in \Z_+$,
\begin{equation*}
\lim_{N\to\infty}q_k \Bigl( N-1,\frac{\meandeg}{N} \Bigr) = \pi_k, 
\lim_{N\to\infty}Q_k \Bigl(N-1,\frac{\meandeg}{N} \Bigr) = \Pi_k.
\end{equation*}
Letting $N \to \infty$ in \eqref{eqn:PreLimit}, substituting these limiting distribution results, using continuity of the generating function $G_D$ and the definition of $f(s)$ in \eqref{eqn:ConProbGreedy4}, it follows that for any sender-receiver pair $(u,v)$, 
\begin{equation}
\label{eqn:LimEdgeProbConDeg}
\lim_{N\to \infty} \E[\xi_{uv}(-\infty) \mid u \in \nbrh_v, \deg(v)=s] = f(s).
\end{equation}
Recall from Remark~\ref{rem:EdgeProb} that $\deg(v)$ has a $\Bin(N, {\meandeg}/{N})$ distribution, which converges to a $\Pois(\meandeg)$ distribution as $N \to \infty$ when $\meandeg$ is fixed.
It follows that for any $\epsilon>0$, the probability that $\deg(v)$ exceeds
$$\min \set{\sqrt{N(\meandeg)^2/(\E D^2)^2}, \sqrt[3]{\epsilon N(\meandeg)^2/(\E D^2)^2} }$$
decays exponentially in $\sqrt[3]{N}$.
Thus, substituting \eqref{eqn:LimEdgeProbConDeg} into Lemma~\ref{lem:ProbDisconn1}, we obtain for $\epsilon > 0$ that
\begin{equation*}
\limsup_{N\to \infty} \E[\psi_v(-\infty)\mid \deg(v)=s] \le e^\epsilon \Bigl( 1-f(s) \Bigr)^s.
\end{equation*}
As this bound is the same for all receivers $v\in V$ and holds for arbitrary $\epsilon>0$, the claim follows by letting $\epsilon \to 0$ and summing over receiver degrees with the limiting $\Pois(\meandeg)$ weights.
\end{IEEEproof}

\begin{cor}
\label{cor:RandRxGreedy}
When the out-degree distribution is deterministic or Poisson, $f(s)$ simplifies as follows. 
\EQ{
f(s) = \begin{cases} 
\frac{\bar{\Pi}_{s-2}^{\meandeg}-\bar{\Pi}_{s-1}^{\meandeg}}{\meandeg \pi_{s-1}},& D \equiv \meandeg,\\
 \frac{e^{-\meandeg \Pi_{s-2}}-e^{-\meandeg \Pi_{s-1}}}{\meandeg \pi_{s-1}},& D\sim \Pois(\meandeg).
\end{cases}
}
\end{cor}

\section{Numerical evaluations}
\label{sec:numerical}

We briefly recall key terminology.
The input to the matching problem is a bipartite \emph{feasible graph}, whose edges $(u,v)$ denote which senders $u\in U$ have messages for which receivers $v\in V$.
The feasible graph may be subjected to thinning (Definition \ref{defn:Thinning}) to obtain an \emph{intention graph}, which is the input to the DB$(\alpha)$ algorithm.

In Section \ref{sec:Analysis}, we assumed that the intention graph is a $D$-out random graph, i.e., that sender degrees are \iid and, conditional on its degree, each sender chooses its neighbors independently and uniformly at random. 
Under this assumption, we derived an expression for the mean matching size under the DB$(0)$ algorithm and a lower bound under the DB$(-\infty)$ algorithm,
in Theorems \ref{thm:RandRxUnif} and \ref{thm:RandRxGreedy} respectively. 
We simplified these results, obtaining closed form results in Corollaries \ref{cor:RandRxUnif} and \ref{cor:RandRxGreedy}, for sender degree $D$ that were either deterministic or had a Binomial distribution.
We now compare these analytical results with empirical results from numerical simulations.

All our simulations are carried out for a system with an equal number $N=144$ of senders and receivers. 
Edges are generated according to the $D$-out random graph model to obtain a feasible graph, which may (or may not) be thinned to obtain the intention graph. 
The DB($\alpha$) algorithm is applied to the intention graph to obtain a matching. 
Each plot in this section is based on $1{,}000$ replicates.

In Fig.~\ref{fig:03}, we depict results for DB($0$) and DB($-\infty$) algorithms applied to bipartite Erd\H{o}s-R\'enyi intention graphs, which have $\Bin(N,p)$ out-degree distributions. 
We plot the fraction of nodes that are successfully matched against the mean degree, $\meandeg = Np$. 
We plot the theoretical prediction from Corollaries \ref{cor:RandRxUnif} and \ref{cor:RandRxGreedy} as well as the mean and first and third quartiles from $1{,}000$ simulations. 

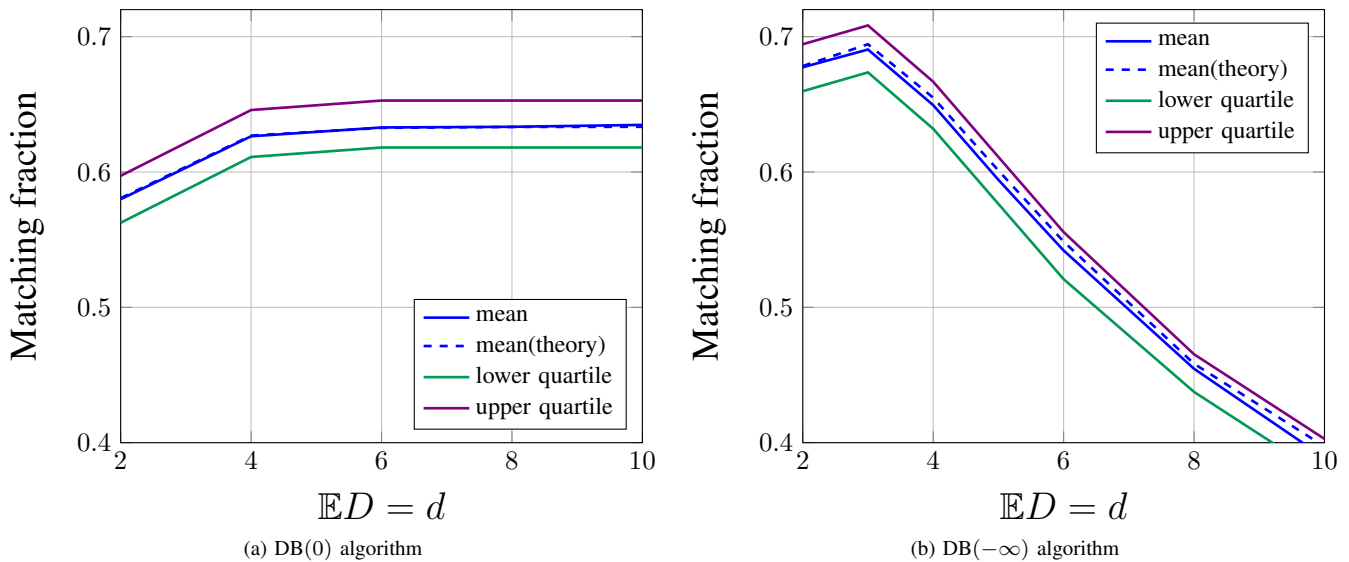
\begin{figure}[ht]
\centering
\subfloat[DB$(0)$ algorithm\label{fig:03a}]{\resizebox{0.5\linewidth}{!}{\pgfplotsset{every axis/.append style={label style={font=\Large}}}
\begin{tikzpicture}

\begin{axis}[
xlabel={ $\meandeg = d$},
ylabel={Matching fraction},
xmajorgrids,
ymajorgrids,
ymin= 0.4, ymax=0.72,
xmin=2, xmax=10,
legend entries={mean, mean(theory),lower quartile, upper quartile
},
legend style = {legend pos = south east, nodes=right, font=\small},
title = { }
]

 \addplot[
 color= blue,
 line width=1pt,
 every mark/.append style={solid},
 ]
 coordinates {
 (2, 0.5799)
(4, 0.6264)
(6, 0.6329)
(8, 0.6334)
(10, 0.6348)
  };			

 \addplot[
 color= blue,
dashed,
 line width=1pt,
 every mark/.append style={solid},
 ]
 coordinates {
 (2, 0.5807)
(4, 0.6269)
(6, 0.6326)
(8, 0.6333)
(10, 0.6334)
  };				
 
\addplot[
color= ForestGreen,
line width=1pt,
every mark/.append style={solid},
]
coordinates {
(2, 0.5625)
(4, 0.6111)
(6, 0.6181)
(8, 0.6181)
(10, 0.6181)
 };

\addplot[
color= violet,
line width=1pt,
every mark/.append style={solid},
]
coordinates {
(2, 0.5972)
(4, 0.6458)
(6, 0.6528)
(8, 0.6528)
(10, 0.6528)
 };

\end{axis}
\end{tikzpicture}}}
\subfloat[DB$(-\infty)$ algorithm\label{fig:03b}]{\resizebox{0.5\linewidth}{!}{\pgfplotsset{every axis/.append style={label style={font=\Large}}}
\begin{tikzpicture}

\begin{axis}[
xlabel={ $\meandeg=d$},
ylabel={Matching fraction},
xmajorgrids,
ymajorgrids,
ymin=.4, ymax=.72,
xmin=2, xmax=10,
legend entries={mean, mean(theory), lower quartile, upper quartile
},
legend style = {legend pos = north east, nodes=right, font=\small},
title = { }
]

 \addplot[
 color= blue,
 line width=1pt,
 every mark/.append style={solid},
 ]
 coordinates {
 (2, 0.6774)
(3, 0.6905)
(4, 0.6494)
(5, 0.5942)
(6, 0.5420)
(8, 0.4546)
(10, 0.3889)
  };			

 \addplot[
 color= blue,
dashed,
 line width=1pt,
 every mark/.append style={solid},
 ]
 coordinates {
 (2, 0.6783)
(3, 0.6944)
(4, 0.6549)
(5, 0.6011)
(6, 0.5486)
(8, 0.4583)
(10, 0.3972)
  };				
 
\addplot[
color= ForestGreen,
line width=1pt,
every mark/.append style={solid},
]
coordinates {
(2, 0.6597)
(3, 0.6736)
(4, 0.6319)
(5, 0.5764)
(6, 0.5208)
(8, 0.4375)
(10, 0.3750)
 };

\addplot[
color= violet,
line width=1pt,
every mark/.append style={solid},
]
coordinates {
(2, 0.6944)
(3, 0.7083)
(4, 0.6667)
(5, 0.6111)
(6, 0.5556)
(8, 0.4653)
(10, 0.4028)
 };

\end{axis}
\end{tikzpicture}}}
\caption{
Fraction of nodes matched versus mean degree for $D$-out random graphs with $D\sim \Bin(N,d/N)$.
Theoretical predictions closely track the empirical mean, and the tight inter-quartile range over $1{,}000$ runs at $N=144$ confirms good prediction.
}
\label{fig:03}
\end{figure}
For DB$(0)$, the theoretical predictions from Corollary~\ref{cor:RandRxUnif} are exact for finite $N$ and match the empirical results up to simulation noise.
For DB$(-\infty)$, the lower bound from Corollary~\ref{cor:RandRxGreedy}, derived in the limiting regime $N\to\infty$, already provides an excellent approximation at $N=144$.
Furthermore, the inter-quartile range is small, i.e., the empirical matching fraction concentrates close to its mean. 
This is of great practical relevance, as it shows that the mean performance prediction can (nearly) be treated as a performance guarantee that will hold across most random problem instances. 
These findings are echoed in Fig. \ref{fig:04}, which depicts results for random bipartite graphs with deterministic sender degrees. 
Again, we plot the fraction of matched nodes against mean node degree, and compare the theoretical prediction from Corollaries \ref{cor:RandRxUnif} and \ref{cor:RandRxGreedy} with the empirical mean and first and third quartiles from $1{,}000$ simulations.

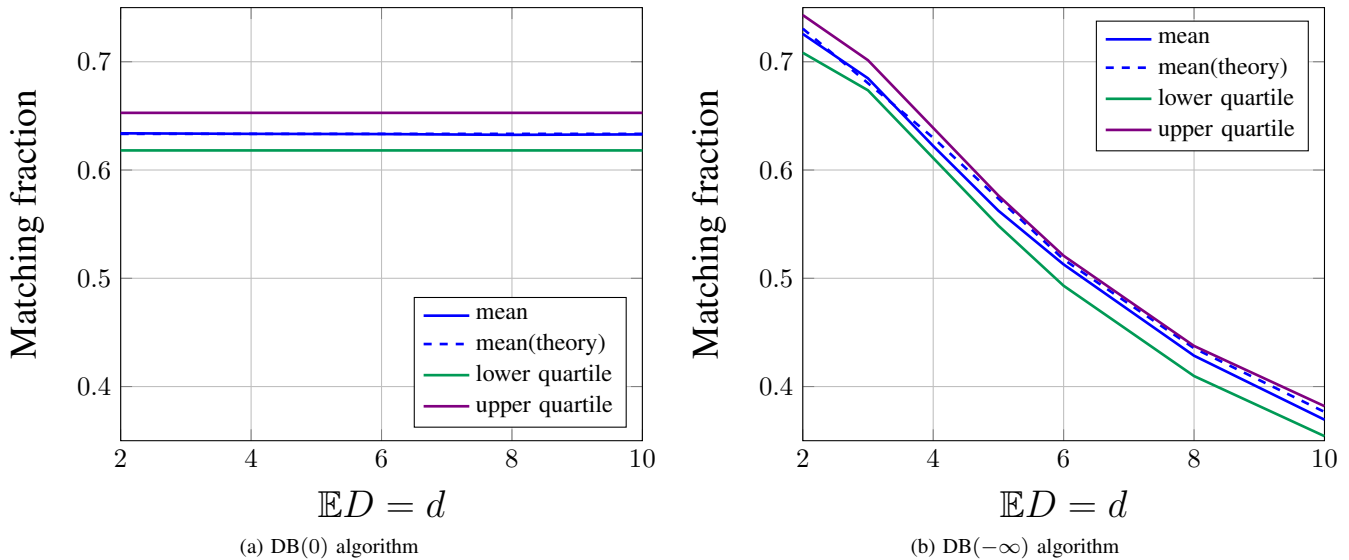
\begin{figure}[ht]
\centering
\subfloat[DB$(0)$ algorithm \label{fig:04a}]{\resizebox{0.5\linewidth}{!}{\pgfplotsset{every axis/.append style={label style={font=\Large}}}
\begin{tikzpicture}

\begin{axis}[
xlabel={ $\meandeg = d$},
ylabel={Matching fraction},
xmajorgrids,
ymajorgrids,
ymin= 0.35,
ymax=0.75,
xmin=2,
xmax=10,
,
legend entries={mean, mean(theory),lower quartile, upper quartile
},
legend style = {legend pos = south east, nodes=right, font=\small},
title = { }
]

 \addplot[
 color= blue,
 line width=1pt,
 every mark/.append style={solid},
 ]
 coordinates {
 (2, 0.6339)
(4, 0.6334)
(6, 0.6332)
(8, 0.6325)
(10, 0.6329)
  };			

 \addplot[
 color= blue,
dashed,
 line width=1pt,
 every mark/.append style={solid},
 ]
 coordinates {
 (2, 0.6334)
(4, 0.6334)
(6, 0.6334)
(8, 0.6334)
(10, 0.6334)
  };				
 
\addplot[
color= ForestGreen,
line width=1pt,
every mark/.append style={solid},
]
coordinates {
(2, 0.6181)
(4, 0.6181)
(6, 0.6181)
(8, 0.6181)
(10, 0.6181)
 };

\addplot[
color= violet,
line width=1pt,
every mark/.append style={solid},
]
coordinates {
(2, 0.6528)
(4, 0.6528)
(6, 0.6528)
(8, 0.6528)
(10, 0.6528)
 };

\end{axis}
\end{tikzpicture}}}
\subfloat[DB$(-\infty)$ algorithm\label{fig:04b}]{\resizebox{0.5\linewidth}{!}{\pgfplotsset{every axis/.append style={label style={font=\Large}}}
\begin{tikzpicture}

\begin{axis}[
xlabel={ $\meandeg=d$},
ylabel={Matching fraction},
xmajorgrids,
ymajorgrids,
ymin= 0.35,
ymax=0.75,
xmin=2,
xmax=10,
legend entries={mean, mean(theory),lower quartile, upper quartile
},
legend style = {legend pos = north east, nodes=right, font=\small},
title = { }
]

 \addplot[
 color= blue,
 line width=1pt,
 every mark/.append style={solid},
 ]
 coordinates {
 (2, 0.7257)
(3, 0.6847)
(4, 0.6222)
(5, 0.5624)
(6, 0.5128)
(8, 0.4285)
(10, 0.3694)
  };			

 \addplot[
 color= blue,
dashed,
 line width=1pt,
 every mark/.append style={solid},
 ]
 coordinates {
 (2, 0.7306)
(4, 0.6298)
(6, 0.5174)
(8, 0.4354)
(10, 0.3766)
  };				
 
\addplot[
color= ForestGreen,
line width=1pt,
every mark/.append style={solid},
]
coordinates {
(2, 0.7083)
(3, 0.6736)
(4, 0.6111)
(5, 0.5486)
(6, 0.4931)
(8, 0.4097)
(10, 0.3542)
 };

\addplot[
color= violet,
line width=1pt,
every mark/.append style={solid},
]
coordinates {
(2, 0.7431)
(3, 0.7014)
(4, 0.6389)
(5, 0.5764)
(6, 0.5208)
(8, 0.4375)
(10, 0.3819)
 };

\end{axis}
\end{tikzpicture}}}
\caption{
Fraction of nodes matched versus mean degree for $D$-out random graphs with deterministic $D=d$.
The empirical mean over $1{,}000$ runs at $N=144$ agrees well with theory, and the narrow inter-quartile range attests to good concentration.
}
\label{fig:04}
\end{figure}

Observe from Fig.~\ref{fig:03} that, while the matching fraction increases slightly with mean degree before saturating when $\alpha=0$, it decreases with mean degree when $\alpha=-\infty$, except for a brief initial increase. 
This difference is more pronounced for random graphs with deterministic sender degrees, shown in Fig.~\ref{fig:04}.
Here, the matching fraction under DB$(0)$ remains insensitive to the out-degree, consistent with Corollary~\ref{cor:RandRxUnif}, while it decreases with out-degree under DB$(-\infty)$.

These observations are explained as follows.
Under DB$(-\infty)$, each sender restricts its grant to a neighbor of minimum degree.
As the mean sender degree increases, the receiver degree distribution concentrates at higher values, so the few receivers with the smallest degrees attract grants from disproportionately many senders, leading to collisions.
Under DB$(0)$, grants are spread uniformly across all neighbors, avoiding this bottleneck but forgoing the advantage of targeting under-demanded receivers at low mean degrees. 
This motivates us to explore the optimal choice of the DB algorithm parameter, $\alpha$.

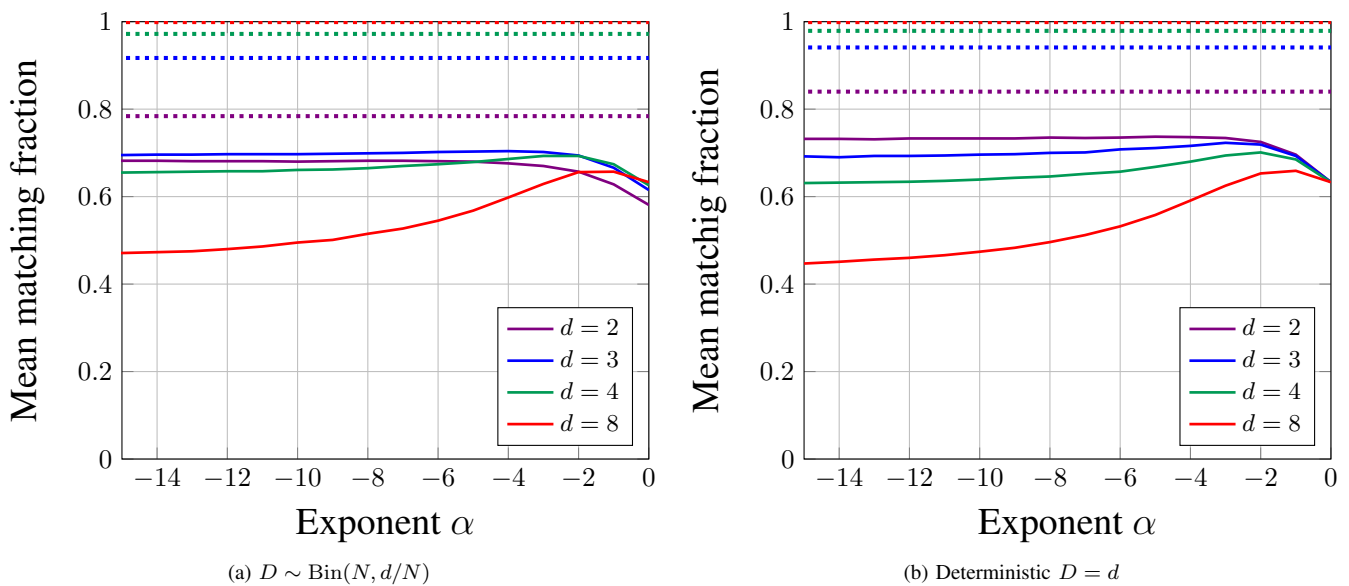
\begin{figure}[ht]
\centering
\subfloat[$D\sim \Bin(N,d/N)$ \label{fig:05a}]{\resizebox{0.5\linewidth}{!}{\pgfplotsset{every axis/.append style={label style={font=\Large}}}
\begin{tikzpicture}

\begin{axis}[
xlabel={Exponent $\alpha$},
ylabel={Mean matching fraction},
xmajorgrids,
ymajorgrids,
ymin= 0,
ymax=1,
xmin=-15,
xmax=0,
legend entries={ $d =2$, ,$d =3$, ,$d=4$,,$d =8$,
},
legend style = {legend pos = south east, nodes=right, font=\small},
title = { }
]

\addplot[
color= violet,
line width=1pt,
every mark/.append style={solid},
]
coordinates {
(-20, 0.682)
(-19, 0.682)
(-18, 0.681)
(-17, 0.680)
(-16, 0.681)
(-15, 0.682)
(-14, 0.682)
(-13, 0.681)
(-12, 0.681)
(-11, 0.681)
(-10, 0.680)
(-9, 0.681)
(-8, 0.682)
(-7, 0.682)
(-6, 0.681)
(-5, 0.680)
(-4, 0.676)
(-3, 0.670)
(-2, 0.657)
(-1, 0.628)
(0, 0.581)
};			

\addplot[
color= violet,
dotted,
line width=1.5pt,
every mark/.append style={solid},
]
coordinates {
(-20, 0.784)
(-19, 0.784)
(-18, 0.784)
(-17, 0.784)
(-16, 0.784)
(-15, 0.784)
(-14, 0.784)
(-13, 0.784)
(-12, 0.784)
(-11, 0.784)
(-10, 0.784)
(-9, 0.784)
(-8, 0.784)
(-7, 0.784)
(-6, 0.784)
(-5, 0.784)
(-4, 0.784)
(-3, 0.784)
(-2, 0.784)
(-1, 0.784)
(0, 0.784)
};

\addplot[
color= blue,
line width=1pt,
every mark/.append style={solid},
]
coordinates {
(-20, 0.695)
(-19, 0.695)
(-18, 0.696)
(-17, 0.694)
(-16, 0.695)
(-15, 0.695)
(-14, 0.696)
(-13, 0.696)
(-12, 0.697)
(-11, 0.697)
(-10, 0.697)
(-9, 0.698)
(-8, 0.699)
(-7, 0.700)
(-6, 0.702)
(-5, 0.703)
(-4, 0.704)
(-3, 0.702)
(-2, 0.694)
(-1, 0.666)
(0, 0.615)
 };

\addplot[
color= blue,
dotted,
line width=1.5pt,
every mark/.append style={solid},
]
coordinates {
(-20, 0.917)
(-19, 0.917)
(-18, 0.917)
(-17, 0.917)
(-16, 0.917)
(-15, 0.917)
(-14, 0.917)
(-13, 0.917)
(-12, 0.917)
(-11, 0.917)
(-10, 0.917)
(-9, 0.917)
(-8, 0.917)
(-7, 0.917)
(-6, 0.917)
(-5, 0.917)
(-4, 0.917)
(-3, 0.917)
(-2, 0.917)
(-1, 0.917)
(0, 0.917)
};

\addplot[
color=ForestGreen,
line width=1pt,
every mark/.append style={solid},
]
coordinates {
(-20, 0.654)
(-19, 0.655)
(-18, 0.655)
(-17, 0.655)
(-16, 0.656)
(-15, 0.655)
(-14, 0.656)
(-13, 0.657)
(-12, 0.658)
(-11, 0.658)
(-10, 0.661)
(-9, 0.662)
(-8, 0.665)
(-7, 0.670)
(-6, 0.674)
(-5, 0.679)
(-4, 0.686)
(-3, 0.693)
(-2, 0.693)
(-1, 0.674)
(0, 0.626)
 };

\addplot[
color=ForestGreen,
dotted,
line width=1.5pt,
every mark/.append style={solid},
]
coordinates {
(-20, 0.972)
(-19, 0.972)
(-18, 0.972)
(-17, 0.972)
(-16, 0.972)
(-15, 0.972)
(-14, 0.972)
(-13, 0.972)
(-12, 0.972)
(-11, 0.972)
(-10, 0.972)
(-9, 0.972)
(-8, 0.972)
(-7, 0.972)
(-6, 0.972)
(-5, 0.972)
(-4, 0.972)
(-3, 0.972)
(-2, 0.972)
(-1, 0.972)
(0, 0.972)
 };

 \addplot[
color= red,
line width=1pt,
every mark/.append style={solid},
]
coordinates {
(-20, 0.462)
(-19, 0.464)
(-18, 0.464)
(-17, 0.467)
(-16, 0.468)
(-15, 0.471)
(-14, 0.473)
(-13, 0.475)
(-12, 0.480)
(-11, 0.486)
(-10, 0.495)
(-9, 0.501)
(-8, 0.515)
(-7, 0.527)
(-6, 0.545)
(-5, 0.568)
(-4, 0.598)
(-3, 0.629)
(-2, 0.656)
(-1, 0.657)
(0, 0.633)
 };

 \addplot[
color= red,
dotted,
line width=1.5pt,
every mark/.append style={solid},
]
coordinates {
(-20, 0.999)
(-19, 0.999)
(-18, 0.999)
(-17, 0.999)
(-16, 0.999)
(-15, 0.999)
(-14, 0.999)
(-13, 0.999)
(-12, 0.999)
(-11, 0.999)
(-10, 0.999)
(-9, 0.999)
(-8, 0.999)
(-7, 0.999)
(-6, 0.999)
(-5, 0.999)
(-4, 0.999)
(-3, 0.999)
(-2, 0.999)
(-1, 0.999)
(0, 0.999)
 };

\end{axis}
\end{tikzpicture}}}
\subfloat[Deterministic $D = d$ \label{fig:05b}]{\resizebox{0.5\linewidth}{!}{\pgfplotsset{every axis/.append style={label style={font=\Large}}}
\begin{tikzpicture}
\begin{axis}[
xlabel={Exponent $\alpha$},
ylabel={Mean matchig fraction},
xmajorgrids,
ymajorgrids,
ymin=0,
ymax=1,
xmin=-15,
xmax=0,
legend entries={$d=2$,,$d=3$,,$d=4$, ,$d=8$,
},
legend style = {legend pos =  south east, nodes=right, font=\small},
title = { }
]

\addplot[
color= violet,
line width=1pt,
every mark/.append style={solid},
]
coordinates {
(-20, 0.730)
(-19, 0.731)
(-18, 0.731)
(-17, 0.732)
(-16, 0.731)
(-15, 0.732)
(-14, 0.732)
(-13, 0.731)
(-12, 0.733)
(-11, 0.733)
(-10, 0.733)
(-9, 0.733)
(-8, 0.735)
(-7, 0.734)
(-6, 0.735)
(-5, 0.737)
(-4, 0.736)
(-3, 0.734)
(-2, 0.725)
(-1, 0.696)
(0, 0.633)
 };

\addplot[
color= violet,
dotted,
line width=1.5pt,
every mark/.append style={solid},
]
coordinates {
(-20, 0.840)
(-19, 0.840)
(-18, 0.840)
(-17, 0.840)
(-16, 0.840)
(-15, 0.840)
(-14, 0.840)
(-13, 0.840)
(-12, 0.840)
(-11, 0.840)
(-10, 0.840)
(-9, 0.840)
(-8, 0.840)
(-7, 0.840)
(-6, 0.840)
(-5, 0.840)
(-4, 0.840)
(-3, 0.840)
(-2, 0.840)
(-1, 0.840)
(0, 0.840)
};

\addplot[
color= blue,
line width=1pt,
every mark/.append style={solid},
]
coordinates {
(-20, 0.691)
(-19, 0.690)
(-18, 0.691)
(-17, 0.690)
(-16, 0.693)
(-15, 0.692)
(-14, 0.690)
(-13, 0.693)
(-12, 0.693)
(-11, 0.694)
(-10, 0.696)
(-9, 0.697)
(-8, 0.700)
(-7, 0.701)
(-6, 0.708)
(-5, 0.711)
(-4, 0.716)
(-3, 0.723)
(-2, 0.719)
(-1, 0.693)
(0, 0.633)
 };

\addplot[
color= blue,
dotted,
line width=1.5pt,
every mark/.append style={solid},
]
coordinates {
(-20, 0.941)
(-19, 0.941)
(-18, 0.941)
(-17, 0.941)
(-16, 0.941)
(-15, 0.941)
(-14, 0.941)
(-13, 0.941)
(-12, 0.941)
(-11, 0.941)
(-10, 0.941)
(-9, 0.941)
(-8, 0.941)
(-7, 0.941)
(-6, 0.941)
(-5, 0.941)
(-4, 0.941)
(-3, 0.941)
(-2, 0.941)
(-1, 0.941)
(0, 0.941)
 };

 \addplot[
color=ForestGreen,
line width=1pt,
every mark/.append style={solid},
]
coordinates {
(-20, 0.627)
(-19, 0.628)
(-18, 0.630)
(-17, 0.629)
(-16, 0.631)
(-15, 0.631)
(-14, 0.632)
(-13, 0.633)
(-12, 0.634)
(-11, 0.636)
(-10, 0.639)
(-9, 0.643)
(-8, 0.646)
(-7, 0.652)
(-6, 0.657)
(-5, 0.668)
(-4, 0.680)
(-3, 0.694)
(-2, 0.701)
(-1, 0.685)
(0, 0.633)
 };

 \addplot[
color=ForestGreen,
dotted,
line width=1.5pt,
every mark/.append style={solid},
]
coordinates {
(-20, 0.979)
(-19, 0.979)
(-18, 0.979)
(-17, 0.979)
(-16, 0.979)
(-15, 0.979)
(-14, 0.979)
(-13, 0.979)
(-12, 0.979)
(-11, 0.979)
(-10, 0.979)
(-9, 0.979)
(-8, 0.979)
(-7, 0.979)
(-6, 0.979)
(-5, 0.979)
(-4, 0.979)
(-3, 0.979)
(-2, 0.979)
(-1, 0.979)
(0, 0.979)
 };

 \addplot[
color= red,
line width=1pt,
every mark/.append style={solid},
]
coordinates {
(-20, 0.439)
(-19, 0.440)
(-18, 0.442)
(-17, 0.443)
(-16, 0.444)
(-15, 0.447)
(-14, 0.451)
(-13, 0.456)
(-12, 0.460)
(-11, 0.466)
(-10, 0.474)
(-9, 0.483)
(-8, 0.496)
(-7, 0.512)
(-6, 0.532)
(-5, 0.558)
(-4, 0.591)
(-3, 0.625)
(-2, 0.653)
(-1, 0.659)
(0, 0.633)
 };	

 \addplot[
color= red,
dotted,
line width=1.5pt,
every mark/.append style={solid},
]
coordinates {
(-20, 0.999)
(-19, 0.999)
(-18, 0.999)
(-17, 0.999)
(-16, 0.999)
(-15, 0.999)
(-14, 0.999)
(-13, 0.999)
(-12, 0.999)
(-11, 0.999)
(-10, 0.999)
(-9, 0.999)
(-8, 0.999)
(-7, 0.999)
(-6, 0.999)
(-5, 0.999)
(-4, 0.999)
(-3, 0.999)
(-2, 0.999)
(-1, 0.999)
(0, 0.999)
};
 
\end{axis}
\end{tikzpicture}}}
\caption{Mean matching fraction versus exponent $\alpha$ in $D$-out random graphs. 
Under degree-biased selection, higher mean sender degree can counterintuitively reduce matching performance.
}
\label{fig:05}
\end{figure}

In Fig. \ref{fig:05}, we plot the mean matching fraction as a function of exponent $\alpha \in (-\infty, 0]$, for mean sender degrees $\meandeg \in \set{2,3,4,8}$. 
Fig. \ref{fig:05a} shows results for Erd\H{o}s-R\'enyi random graphs, and Fig.~\ref{fig:05b} for random graphs with deterministic sender degrees. 
We also plot (dotted lines) the mean fraction of nodes matched in the maximum matching in the simulated random graphs, as omniscient upper bounds. 
We see that this quantity increases rapidly with the mean degree; for $\meandeg\ge 4$, the maximum matching is close to perfect. 
This shows that there is a gap between our algorithm and the best centralized algorithm. 
It is an open question whether there is \emph{any} single-round decentralized algorithm that can close this gap.

Next, note that as $\alpha$ becomes more negative, the mean matching size increases for small mean degrees, but first increases and then decreases for large mean degrees. 
These findings are consistent with Fig.~\ref{fig:03} and Fig.~\ref{fig:04}, and highlight that low-degree receivers become bottlenecks if both $-\alpha$ \emph{and} the mean degree are large. 

This observation points to an algorithmic opportunity to improve performance through the preemptive sparsification of the \emph{feasible} communication graph to obtain an \emph{intention} graph. 
Such sparsification yields the added benefit of reducing the number of control messages needed. 
The following are two natural sparsification mechanisms: 
(a) \emph{$\Bern(q)$ thinning}, where each edge is retained with probability $q$, 
independent of other edges, and 
(b) \emph{$\max(k)$ thinning}, where each sender with more than $k$ edges retains exactly $k$ of them, chosen uniformly at random, while senders with $k$ or fewer edges retain all of them. 

\begin{figure}[ht]
\centering
\subfloat[Binomial $D$ with $\max(k)$\label{fig:06a}]{\resizebox{0.48\linewidth}{!}{\pgfplotsset{every axis/.append style={label style={font=\Large}}}
\begin{tikzpicture}

\begin{axis}[
xlabel={Exponent $\alpha$},
ylabel={Mean matching fraction},
xmajorgrids,
ymajorgrids,
ymin=0.5,
ymax=0.75,
xmin=-15,
xmax=0,
legend entries={$(d{,}k) = (4{,}2)$, $(d{,}k) = (4{,}3)$, $(d{,}k) = (4{,}4)$, $(d{,}k) = (8{,}2)$, $(d{,}k) = (8{,}3)$ ,$(d{,}k) = (8{,}4)$,},
legend style = {   
legend pos= south east,
legend cell align=left,
},
title = { }
]

\addplot[
color=ForestGreen,
line width=1pt,
every mark/.append style={solid},
]
coordinates {
(-20, 0.728)
(-19, 0.728)
(-18, 0.727)
(-17, 0.727)
(-16, 0.727)
(-15, 0.728)
(-14, 0.729)
(-13, 0.728)
(-12, 0.727)
(-11, 0.730)
(-10, 0.729)
(-9, 0.729)
(-8, 0.729)
(-7, 0.731)
(-6, 0.730)
(-5, 0.731)
(-4, 0.729)
(-3, 0.726)
(-2, 0.715)
(-1, 0.685)
(0, 0.627)
};

\addplot[
color=ForestGreen,
dashed,
line width=1pt,
every mark/.append style={solid},
]
coordinates {
(-15, 0.710)
(-14, 0.710)
(-13, 0.710)
(-12, 0.711)
(-11, 0.712)
(-10, 0.712)
(-9, 0.713)
(-8, 0.715)
(-7, 0.716)
(-6, 0.718)
(-5, 0.721)
(-4, 0.723)
(-3, 0.723)
(-2, 0.716)
(-1, 0.687)
(0, 0.628)
};

  \addplot[
color=ForestGreen,
dotted,
line width=1.5pt,
every mark/.append style={solid},
]
coordinates {
(-15, 0.686)
(-14, 0.686)
(-13, 0.687)
(-12, 0.687)
(-11, 0.688)
(-10, 0.689)
(-9, 0.691)
(-8, 0.693)
(-7, 0.695)
(-6, 0.699)
(-5, 0.703)
(-4, 0.707)
(-3, 0.711)
(-2, 0.707)
(-1, 0.683)
(0, 0.628)
 };

 \addplot[
color= red, 
line width=1pt,
every mark/.append style={solid},
]
coordinates {
(-20, 0.731)
(-19, 0.732)
(-18, 0.732)
(-17, 0.731)
(-16, 0.731)
(-15, 0.732)
(-14, 0.731)
(-13, 0.732)
(-12, 0.732)
(-11, 0.733)
(-10, 0.733)
(-9, 0.734)
(-8, 0.734)
(-7, 0.736)
(-6, 0.736)
(-5, 0.736)
(-4, 0.737)
(-3, 0.734)
(-2, 0.724)
(-1, 0.695)
(0, 0.633)
 };

  \addplot[
color= red, 
dashed,
line width=1pt,
every mark/.append style={solid},
]
coordinates {
(-15, 0.693)
(-14, 0.693)
(-13, 0.693)
(-12, 0.694)
(-11, 0.695)
(-10, 0.697)
(-9, 0.698)
(-8, 0.700)
(-7, 0.703)
(-6, 0.707)
(-5, 0.712)
(-4, 0.717)
(-3, 0.722)
(-2, 0.719)
(-1, 0.694)
(0, 0.633)
 };

   \addplot[
 color= red, 
 dotted,
 line width=1.5pt,
 every mark/.append style={solid},
 ]
 coordinates {
(-15, 0.635)
(-14, 0.636)
(-13, 0.637)
(-12, 0.639)
(-11, 0.640)
(-10, 0.643)
(-9, 0.646)
(-8, 0.650)
(-7, 0.655)
(-6, 0.662)
(-5, 0.671)
(-4, 0.683)
(-3, 0.695)
(-2, 0.702)
(-1, 0.685)
(0, 0.633)
  };

\end{axis}
\end{tikzpicture}}}
\hfill
\subfloat[Deterministic $D$ with $\Bern(k/d)$\label{fig:06b}]{\resizebox{0.48\linewidth}{!}{\pgfplotsset{every axis/.append style={label style={font=\Large}}}
\begin{tikzpicture}

\begin{axis}[
xlabel={Exponent $\alpha$},
ylabel={Mean matching fraction},
xmajorgrids,
ymajorgrids,
ymin=0.5,
ymax=0.75,
xmin=-15,
xmax=0,
legend entries={$(d{,}k) = (4{,}2)$, $(d{,}k) = (4{,}3)$, $(d{,}k) = (4{,}4)$, $(d{,}k) = (8{,}2)$, $(d{,}k) = (8{,}3)$ ,$(d{,}k) = (8{,}4)$,},
legend style = {   legend pos=south east,
  legend cell align=left,
  font=\small},
title = { }
]

\addplot[
color=ForestGreen,
line width=1pt,
every mark/.append style={solid},
]
coordinates {
(-15, 0.7137598055555557)
(-14, 0.7137937777777778)
(-13, 0.7138407777777778)
(-12, 0.7138953055555555)
(-11, 0.7139654166666665)
(-10, 0.7140501388888889)
(-9, 0.7141307708333335)
(-8, 0.7141887638888889)
(-7, 0.7141464861111112)
(-6, 0.713857375)
(-5, 0.7129656180555555)
(-4, 0.710638625)
(-3, 0.7049797013888889)
(-2, 0.6915815416666666)
(-1, 0.6618358680555555)
(0, 0.609909625)
};

\addplot[
color=ForestGreen,
dashed,
line width=1pt,
every mark/.append style={solid},
]
coordinates {
(-15, 0.7014197638888889)
(-14, 0.701771638888889)
(-13, 0.7022271249999999)
(-12, 0.7028038055555555)
(-11, 0.7035638125)
(-10, 0.7045456666666667)
(-9, 0.7058154791666667)
(-8, 0.7074798263888889)
(-7, 0.7096326319444445)
(-6, 0.7124172569444444)
(-5, 0.7158013680555556)
(-4, 0.7194435972222222)
(-3, 0.7215795555555556)
(-2, 0.7163759375)
(-1, 0.6902482291666666)
(0, 0.6322893263888889)
};

  \addplot[
color=ForestGreen, 
dotted,
line width=1.5pt,
every mark/.append style={solid},
]
coordinates {
(-15, 0.63011375)
(-14, 0.6310480833333334)
(-13, 0.6322285138888889)
(-12, 0.6337209513888889)
(-11, 0.6356239166666667)
(-10, 0.6381001111111111)
(-9, 0.6412925416666667)
(-8, 0.6454705902777778)
(-7, 0.6509555694444445)
(-6, 0.6582200625)
(-5, 0.6677758333333334)
(-4, 0.6798337638888889)
(-3, 0.6930398472222222)
(-2, 0.7006487222222223)
(-1, 0.6852169375)
(0, 0.6336616736111111)
 };

 \addplot[
color= red, 
line width=1pt,
every mark/.append style={solid},
]
coordinates {
(-15, 0.696871826388889)
(-14, 0.6968962361111111)
(-13, 0.69692425)
(-12, 0.6969597152777779)
(-11, 0.696996576388889)
(-10, 0.6970327083333333)
(-9, 0.6970633680555556)
(-8, 0.6970577152777778)
(-7, 0.6969182222222222)
(-6, 0.6966356388888888)
(-5, 0.6955453819444444)
(-4, 0.6929180416666666)
(-3, 0.686887826388889)
(-2, 0.67312925)
(-1, 0.6439052569444443)
(0, 0.5946849722222223)
 };

  \addplot[
color= red, 
dashed,
line width=1pt,
every mark/.append style={solid},
]
coordinates {
(-15, 0.7022277083333333)
(-14, 0.7024754513888889)
(-13, 0.7027885833333333)
(-12, 0.7032092361111111)
(-11, 0.7037560416666666)
(-10, 0.704464138888889)
(-9, 0.7053813819444444)
(-8, 0.7065684236111111)
(-7, 0.7080861736111112)
(-6, 0.7100016736111111)
(-5, 0.7121729097222222)
(-4, 0.7141275208333334)
(-3, 0.7139946666666668)
(-2, 0.7062472638888888)
(-1, 0.6790997986111111)
(0, 0.6247125277777779)
 };

   \addplot[
 color= red, 
 dotted,
 line width=1.5pt,
 every mark/.append style={solid},
 ]
 coordinates {
(-15, 0.6493987986111111)
(-14, 0.6500817430555555)
(-13, 0.6509588402777777)
(-12, 0.6520745763888889)
(-11, 0.6535172291666666)
(-10, 0.6553763263888889)
(-9, 0.6577970972222222)
(-8, 0.6609812013888889)
(-7, 0.665197923611111)
(-6, 0.670780611111111)
(-5, 0.6780988680555555)
(-4, 0.6871457291666667)
(-3, 0.696373534722222)
(-2, 0.6994581736111111)
(-1, 0.6814142083333334)
(0, 0.6320930902777777)
  };

\end{axis}
\end{tikzpicture}}}
\caption{Mean matching fraction versus exponent $\alpha$ in $D$-out random graphs having $\meandeg = d$ with thinning for $k \in \set{2,3,4}$.
Preemptive thinning to low out-degree consistently improves decentralized matching performance, with $\max(2)$ yielding the strongest results.
}
\label{fig:06}
\end{figure}
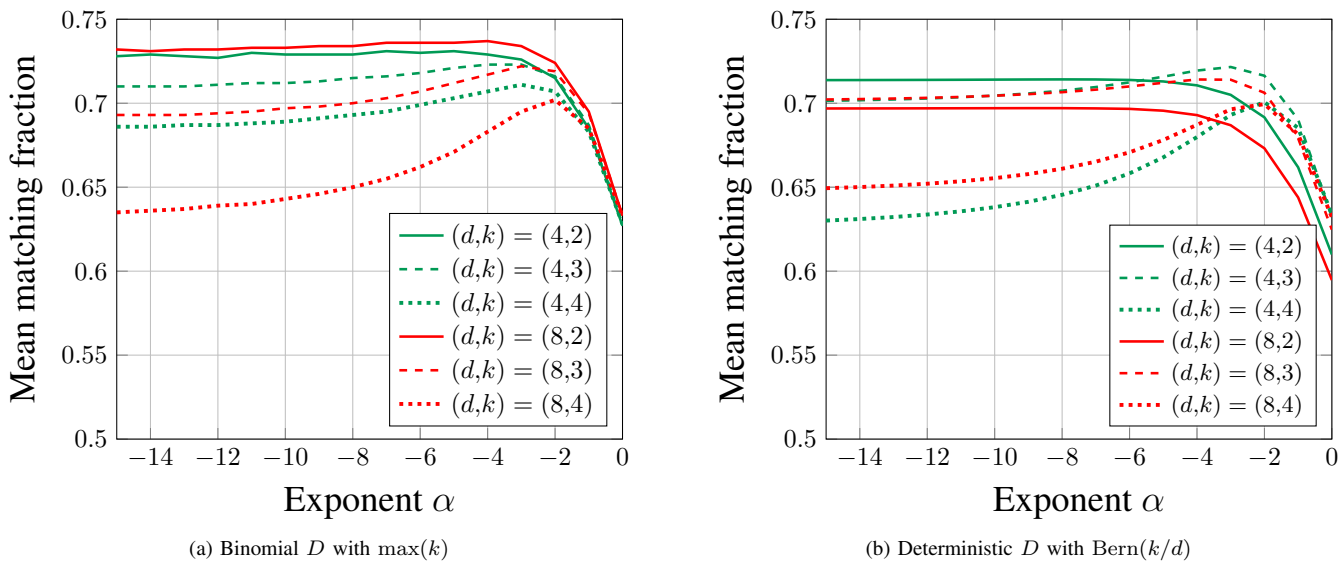

We have depicted the effect of thinning on the size of the matching obtained by the DB($\alpha$) algorithm in Fig.~\ref{fig:06}, for different values of $\alpha$, the mean degree $\meandeg$, and the thinning parameter. 
Fig.~\ref{fig:06a} shows results for max($k$) thinning of Erd\H{o}s-R\'enyi random graphs and Fig.~\ref{fig:06b} for Bernoulli thinning of random graphs with deterministic out-degrees. 
Note that $\Bern(q)$ thinning of a $G(n,n,p)$ graph yields a $G(n,n,pq)$ graph, while $\max(k)$ thinning of a $d$-out random graph yields a $\min(d,k)$-out random graph.
Since these cases reduce to previously considered graph families, we omit the corresponding plots. 
We see from the figure that $\max(2)$ thinning yields the best performance across a range of parameter values. 
An intuitive explanation is that, as seen in Fig.~\ref{fig:05}, the $2$-out random graph achieves the best performance in the class of random graphs with deterministic out-degrees, and $\max(2)$ thinning yields intention graphs that approximate $2$-out random graphs, especially when applied to dense feasible graphs. 

In light of the above, we now compare the performance of the DB($-\infty$) algorithm on random feasible graphs sparsified by $\max(2)$ thinning with the performance of the optimal DB($\alpha$) algorithm on unthinned graphs. 
In order to do so, we numerically obtain the optimal exponent $\alpha^\ast(\meandeg)$ for the values of $\meandeg$ used in Fig.~\ref{fig:05}, for random graphs with Binomial and deterministic out-degree distributions. 
The results are shown in Tables \ref{tab:ER} and \ref{tab:Det}, for Binomial and deterministic out-degrees, respectively, along with the corresponding mean matching fractions. 
The tables also show the performance of the DB($0$) and DB($-\infty$) algorithms for ease of comparison, along with the performance obtained by using $\max(2)$ thinning followed by the DB($-\infty$) algorithm. 

We see that the optimal parameter, $\alpha^\ast$, is non-decreasing in the mean sender degree, $\meandeg$, in both random graph families; as mean degree increases, randomly choosing a neighboring receiver to issue a GRANT becomes progressively preferable to greedily choosing the neighbor of lowest degree. 

We also see from the final column that applying the DB($-\infty$) algorithm to the intention graph obtained by subjecting the feasible graph to $\max(2)$ thinning yields performance that is both robust across mean degrees and that is close to or better than the performance obtained by using DB($\alpha^{\ast}$) on the unthinned graph. This observation leads us to advocate the use of $\max(2)$ thinning followed by DB($-\infty$) matching as a robust algorithm for decentralized matching, with no need for fine-tuning.

\begin{table}[htbp]
\caption{Mean matching fraction under DB$(\alpha)$ for Erd\H{o}s-R\'enyi random graphs}
\label{tab:ER}
\centering
\begin{small}
\begin{tabular}{@{}l ccc c c@{}}
\toprule
& \multicolumn{3}{c}{\textbf{Mean matching DB$(\alpha)$}} & & \textbf{$\max(2) +$} \\
\cmidrule(lr){2-4}
$\meandeg$ & {Uniform} & {Optimal} & {Greedy} & $\alpha^\ast$ & {DB$(-\infty)$} \\
& {$(\alpha = 0)$} & {$(\alpha = \alpha^\ast)$} & {$(\alpha = -\infty)$} & & \\
\midrule
2 & $0.581$ & $0.681$ & $0.681$ & $-\infty$ & $0.678$ \\
3 & $0.618$ & $0.704$ & $0.694$ & $-3.9$ & $0.716$ \\
4 & $0.626$ & $0.695$ & $0.655$ & $-2.4$ & $0.728$ \\
8 & $0.633$ & $0.661$ & $0.455$ & $-1.4$ & $0.731$ \\
\bottomrule
\end{tabular}
\end{small}
\end{table}

\begin{table}[htbp]
\caption{Mean matching fraction under DB$(\alpha)$ for random graphs with deterministic sender degree}
\label{tab:Det}
\centering
\begin{small}
\begin{tabular}{@{}l ccc c c@{}}
\toprule
& \multicolumn{3}{c}{\textbf{Mean matching DB$(\alpha)$}} & & \textbf{$\max(2) +$} \\
\cmidrule(lr){2-4}
$\meandeg$ & {Uniform} & {Optimal} & {Greedy} & $\alpha^\ast$ & {DB$(-\infty)$} \\
& {$(\alpha = 0)$} & {$(\alpha = \alpha^\ast)$} & {$(\alpha = -\infty)$} & & \\
\midrule
2 & $0.633$ & $0.737$ & $0.729$ & $-4.4$ & $0.737$ \\
3 & $0.633$ & $0.722$ & $0.688$ & $-2.7$ & $0.737$ \\
4 & $0.633$ & $0.700$ & $0.625$ & $-2.0$ & $0.737$ \\
8 & $0.633$ & $0.661$ & $0.431$ & $-1.3$ & $0.737$ \\
\bottomrule
\end{tabular}
\end{small}
\end{table}

\section{System simulations for DCNs}
\label{sec:SysSim}

In this section, we evaluate one algorithm from the family proposed above: DB$(-\infty)$ with $\max(2)$ thinning, which we henceforth call the $\twocgs$ (2-choice with greedy selection) algorithm.
This variant offers a favorable balance between matching quality and communication overhead.
We perform packet-level system simulations to evaluate the performance of $\twocgs$ in the DCN context.
The results demonstrate that $\twocgs$ achieves competitive throughput and latency with significantly lower message complexity than existing protocols.

We implement $\twocgs$ within the Yet Another Packet Simulator (YAPS) \cite{NetSysSim} and compare it against two matching-based transport protocols. The first is the one-round version of dcPIM \cite{dcpim_sigcomm2022}, denoted 1r-dcPIM, whose matching algorithm is DB$(0)$ (i.e., uniform selection). The second is iSLIP \cite{islip_tnet1999}, a deterministic round-robin matching algorithm.
Much like PIM, iSLIP operates in three stages: (1) Request, (2) Grant, and (3) Accept.
Senders and receivers maintain their individual fixed round-robin arbitration orders.
In the Grant stage, each receiver selects the highest-priority requesting sender; in the Accept stage, each sender accepts the highest-priority granting receiver.
A receiver's pointer advances only when its grant is accepted; a sender's pointer advances unconditionally.

We consider a two-tier leaf-spine topology consisting of nine Top-of-Rack (ToR) switches, four spine switches, and sixteen hosts per ToR, similar to \cite{dcpim_sigcomm2022}. 
Each host has a $100$~Gbps access link to the ToR, and each ToR is connected to all spine switches with a $400$~Gbps link.
All links have a propagation delay of $200$~ns.
Each switch has a processing latency of $450$~ns and a $500$~kB buffer per output port ($\approx 10 \times$ BDP), where the bandwidth-delay product (BDP) is the maximum amount of data that can be in flight between a sender-receiver pair at any given moment.

Each sender-receiver pair is associated with its own independent Poisson arrival process, which determines the edges of the (evolving) feasible graph. 
At each arrival event, a message size is sampled randomly from a representative workload trace; we evaluate system performance on the following two workload traces: (a) IMC10 \cite{phost_conext15}, and (b) SGD \cite{flowsize_nsdi19}.  
Both are obtained from production networks, and exhibit a heavy-tailed distribution \cite{benson2010network}: \emph{short} messages of size at most one BDP make up the predominant fraction of messages, but the bulk of the traffic volume is due to \emph{long} messages of size greater than one BDP. 
The network load is defined as the ratio of the arrival data rate to the access link data rate. 
We choose arrival rates to generate target network loads within the interval $[0.3, 0.85]$,
representing traffic intensities ranging from light to heavy. 

In all three algorithms, short messages are prioritized and do not require matching before transmission, whereas long messages are transmitted only when the corresponding sender-receiver pairs have been matched, adhering to a receiver-driven grant mechanism for dcPIM \cite{dcpim_sigcomm2022} and $\twocgs$, and a sender-driven grant mechanism for iSLIP \cite{islip_tnet1999}.
Hence, an edge between a sender-receiver pair is added to the feasible graph in our system simulation only when a long message arrives between them.
Control messages are transmitted with the highest priority. 
The distributed matching algorithm is implemented at the hosts, and the leaf-spine switches only forward packets from source to destination. 
The matching and transmission phases are pipelined as in dcPIM \cite{dcpim_sigcomm2022}, to improve utilization and reduce delays. 

The flow completion time (FCT) is the total time taken to transfer a message from the sender to the receiver. 
The optimal FCT is attained when there are no other messages in the system. 

We evaluate our algorithm on the following two metrics: (a) normalized throughput, defined as the ratio of the successful data transfer rate to the access link rate, and (b) normalized FCT, defined as the ratio of the observed FCT to the optimal FCT \cite{dcpim_sigcomm2022}. 
By definition, the successful data transfer rate must be equal to the arrival rate in a stable system. Hence, the stability region is the set of network loads for which the two are equal. 


A matching algorithm that improves the mean matching size increases the number of messages that can be served in each data transmission phase, leading to an enlarged stability region. 
This is illustrated in Figs.~\ref{fig:07} and \ref{fig:08}, where we depict various performance metrics for the $\twocgs$, 1r-dcPIM, and iSLIP algorithms on the IMC10 and SGD workloads, respectively. 
We see from Figs.~\ref{fig:07a} and \ref{fig:08a} that, while all three algorithms achieve the same mean matching fraction at low loads, 1r-dcPIM beats iSLIP at moderate load, while $\twocgs$ beats both at higher loads (on both the IMC10 and SGD workloads). 
Observe from Figs.~\ref{fig:07b} and \ref{fig:08b} for the IMC10 and SGD workloads that the differences in mean matching are reflected in throughput: $\twocgs$ achieves higher normalized throughput than 1r-dcPIM, which in turn beats iSLIP, with the differences becoming more pronounced as the network load increases.
\begin{figure}[ht]
\centering
\subfloat[Mean matching fraction\label{fig:07a}]{\resizebox{.49\linewidth}{!}{\pgfplotsset{every axis ylabel/.append style={font=\large},
	xlabel/.append style={font=\large}}
\pgfplotsset{every tick label/.append style={font=\small}}

\begin{tikzpicture}
\begin{axis}[
xlabel={Network load},
ylabel={Mean matching fraction},
xmajorgrids,
ymajorgrids,
ymin=0.3,
xmin=0.3,
xmax=0.8,
xtick ={0.3, 0.4, 0.5, 0.6, 0.7, 0.8, 0.9},
legend entries={ 1r-dcPIM , $\twocgs$, iSLIP
},
legend style = {legend pos = north west, nodes=right, font=\small},
title = { }
]

\addplot[
color= blue,
line width=1pt,
every mark/.append style={solid},
]
coordinates {
(0.3, 0.335355)
(0.35, 0.391198)
(0.4, 0.447257)
(0.45, 0.503403)
(0.5, 0.559790)
(0.55, 0.612574)
(0.6, 0.630776)
(0.65, 0.632685)
(0.7, 0.633021)
(0.75, 0.633199)
(0.8, 0.633363)
(0.85, 0.633334)
};

  \addplot[
color= red,
line width=1pt,
every mark/.append style={solid},
]
coordinates {
(0.3, 0.335509)
(0.35, 0.392193)
(0.4, 0.449032)
(0.45, 0.505927)
(0.5, 0.562972)
(0.55, 0.619934)
(0.6, 0.675812)
(0.65, 0.715278)
(0.7, 0.729167)
(0.75, 0.729167)
(0.8, 0.730463)
(0.85, 0.730692)
 };

  \addplot[
color= ForestGreen,
line width=1pt,
every mark/.append style={solid},
]
coordinates {
(0.3, 0.335594)
(0.4, 0.447361)
(0.5, 0.523611)
(0.6, 0.522917)
(0.7, 0.522917)
(0.8, 0.526389)
(0.9, 0.530556)
 };

\end{axis}
\end{tikzpicture}}}
\subfloat[Normalized throughput\label{fig:07b}]{\resizebox{0.49\linewidth}{!}{\pgfplotsset{every axis ylabel/.append style={font=\large},
	xlabel/.append style={font=\large}}
\pgfplotsset{every tick label/.append style={font=\small}}
\begin{tikzpicture}
\begin{axis}[
xlabel={Network load},
ylabel={Normalized throughput},
xmajorgrids,
ymajorgrids,
ymin=0.3,
xmin=0.3,
xmax=0.8,
xtick ={0.3, 0.4, 0.5, 0.6, 0.7, 0.8, 0.9},
legend entries={ 1r-dcPIM , $\twocgs$, iSLIP
},
legend style = {legend pos = north west, nodes=right, font=\small},
title = { }
]

\addplot[
color= blue,
line width=1pt,
every mark/.append style={solid},
]
coordinates {
(0.3,0.2998)
(0.35,0.3497)
(0.4,0.3995)
(0.45,0.4491)
(0.5,0.498)
(0.55,0.5408)
(0.6,0.5543)
(0.65,0.5551)
(0.7,0.5542)
(0.75,0.5536)
(0.8,0.5533)
(0.85,0.5527)
};

  \addplot[
color= red,
line width=1pt,
every mark/.append style={solid},
]
coordinates {
(0.3,0.2999)
(0.35,0.3498)
(0.4,0.3996)
(0.45,0.4494)
(0.5,0.4991)
(0.55,0.5483)
(0.6,0.5956)
(0.65,0.63)
(0.7,0.65)
(0.75,0.65)
(0.8,0.65)
(0.85,0.65)
 };

  \addplot[
color= ForestGreen,
line width=1pt,
every mark/.append style={solid},
]
coordinates {
(0.3,0.2998)
(0.4,0.3995)
(0.5,0.5)
(0.6,0.5)
(0.7,0.5)
(0.8,0.5)
(0.9,0.5)
 };

\end{axis}
\end{tikzpicture}}}\\
\subfloat[Short message FCT\label{fig:07c}]{\resizebox{.49\linewidth}{!}{\pgfplotsset{every axis ylabel/.append style={font=\large},
	xlabel/.append style={font=\large}}
\pgfplotsset{every tick label/.append style={font=\small}}

\begin{tikzpicture}
\begin{axis}[
xlabel={Network load},
ylabel={Normalized FCT (mean)},
xmajorgrids,
ymajorgrids,
xmin=0.3,
xmax=0.7,
xtick ={0.3, 0.4, 0.5, 0.6, 0.7, 0.8, 0.9},
legend entries={ 1r-dcPIM, $\twocgs$, iSLIP
},
legend style = {legend pos = north west, nodes=right, font=\small},
title = { }
]

\addplot[
color= blue,
line width=1pt,
every mark/.append style={solid},
]
coordinates {
(0.3,1.0581)
(0.35,1.0713)
(0.4,1.0907)
(0.45,1.1232)
(0.5,1.1942)
(0.55,1.463)
(0.6,2.1126)
(0.63,2.4571)
(0.64,2.5755)
(0.65,2.6613)
(0.7,3.0307)
};

  \addplot[
color= red,
line width=1pt,
every mark/.append style={solid},
]
coordinates {
(0.3,1.0585)
(0.35,1.0689)
(0.4,1.0817)
(0.45,1.0986)
(0.5,1.1236)
(0.55,1.1675)
(0.6,1.2807)
(0.63,1.4462)
(0.64,1.515)
(0.65,1.609)
(0.7,2.0233)
 };

  \addplot[
color= ForestGreen,
line width=1pt,
every mark/.append style={solid},
]
coordinates {
(0.3,1.0618)
(0.4,1.0975)
(0.5,2.4624)
(0.6,4.9928)
(0.7,6.9026)
(0.8,1.1120)
(0.9, 1.1300)
 };

\end{axis}
\end{tikzpicture}}}
\subfloat[Long message FCT\label{fig:07d}]{\resizebox{0.49\linewidth}{!}{\pgfplotsset{every axis ylabel/.append style={font=\large},
	xlabel/.append style={font=\large}}
\pgfplotsset{every tick label/.append style={font=\small}}

\begin{tikzpicture}
\begin{axis}[
xlabel={Network load},
ylabel={Normalized FCT (mean)},
xmajorgrids,
ymajorgrids,
xmin=0.3,
xmax=0.7,
xtick ={0.3, 0.4, 0.5, 0.6, 0.7, 0.8, 0.9},
legend entries={ 1r-dcPIM, $\twocgs$, iSLIP
},
legend style = {legend pos = north west, nodes=right, font=\small},
title = { }
]

\addplot[
color= blue,
line width=1pt,
every mark/.append style={solid},
]
coordinates {
(0.3,2.9985)
(0.35,3.6584)
(0.4,4.7085)
(0.45,6.6228)
(0.5,11.1361)
(0.55,29.2836)
(0.6,76.458)
(0.63,102.4977)
(0.64,110.1728)
(0.65,117.4139)
(0.7,145.7684)
};

  \addplot[
color= red,
line width=1pt,
every mark/.append style={solid},
]
coordinates {
(0.3,2.6779)
(0.35,3.0941)
(0.4,3.6601)
(0.45,4.492)
(0.5,5.8381)
(0.55,8.5052)
(0.6,15.8781)
(0.63,27.2404)
(0.64,32.1605)
(0.65,38.2875)
(0.7,70.6772)
 };

  \addplot[
color= ForestGreen,
line width=1pt,
every mark/.append style={solid},
]
coordinates {
(0.3,2.9567)
(0.4,4.7652)
(0.5,68.1021)
(0.6,173.43)
(0.7,244.58)
(0.8,278)
 };
 
\end{axis}
\end{tikzpicture}}}\\
\caption{Performance comparison of $\twocgs$, 1r-dcPIM, and iSLIP for IMC10 workload in leaf-spine architecture.}
\label{fig:07}
\end{figure}
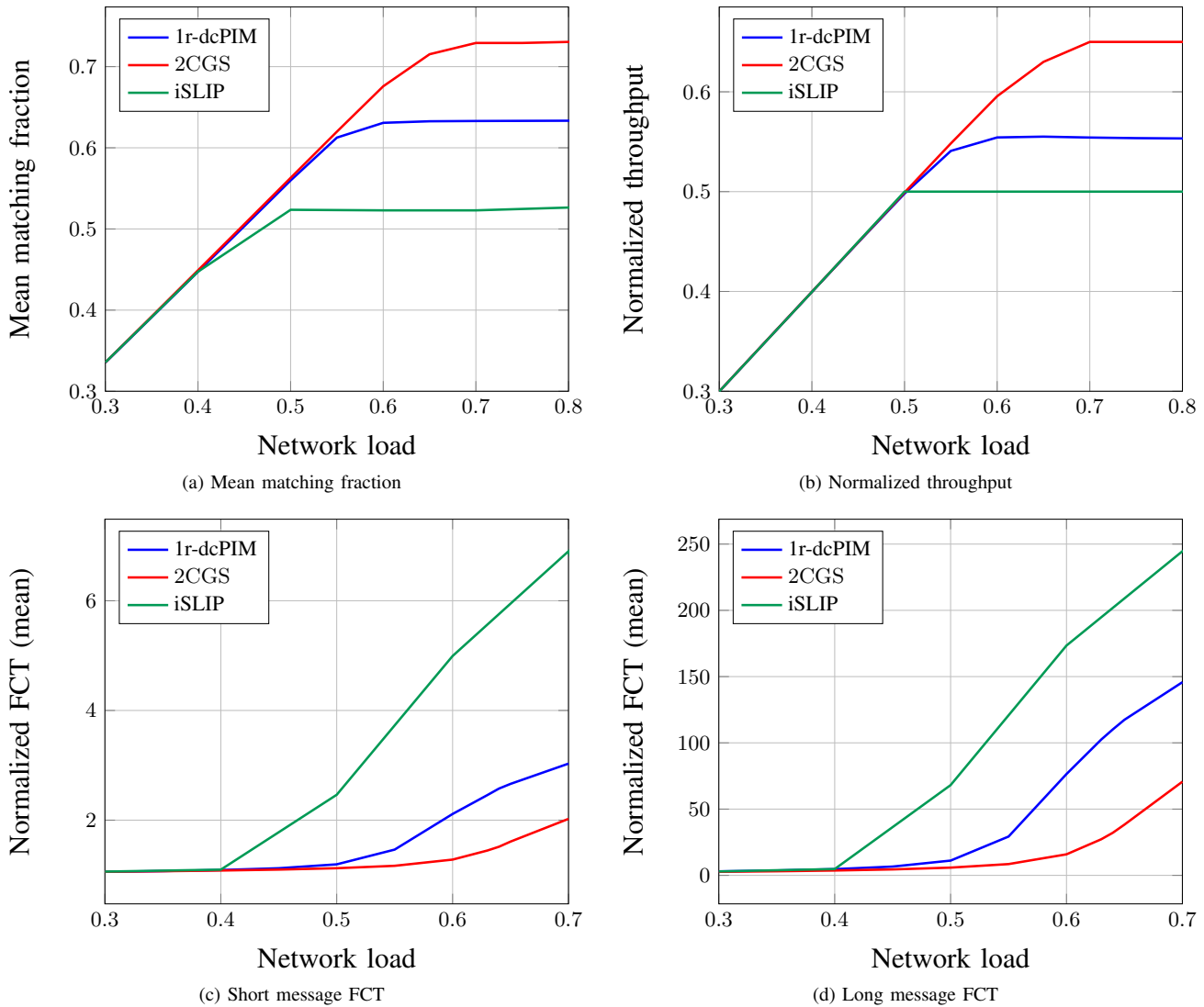
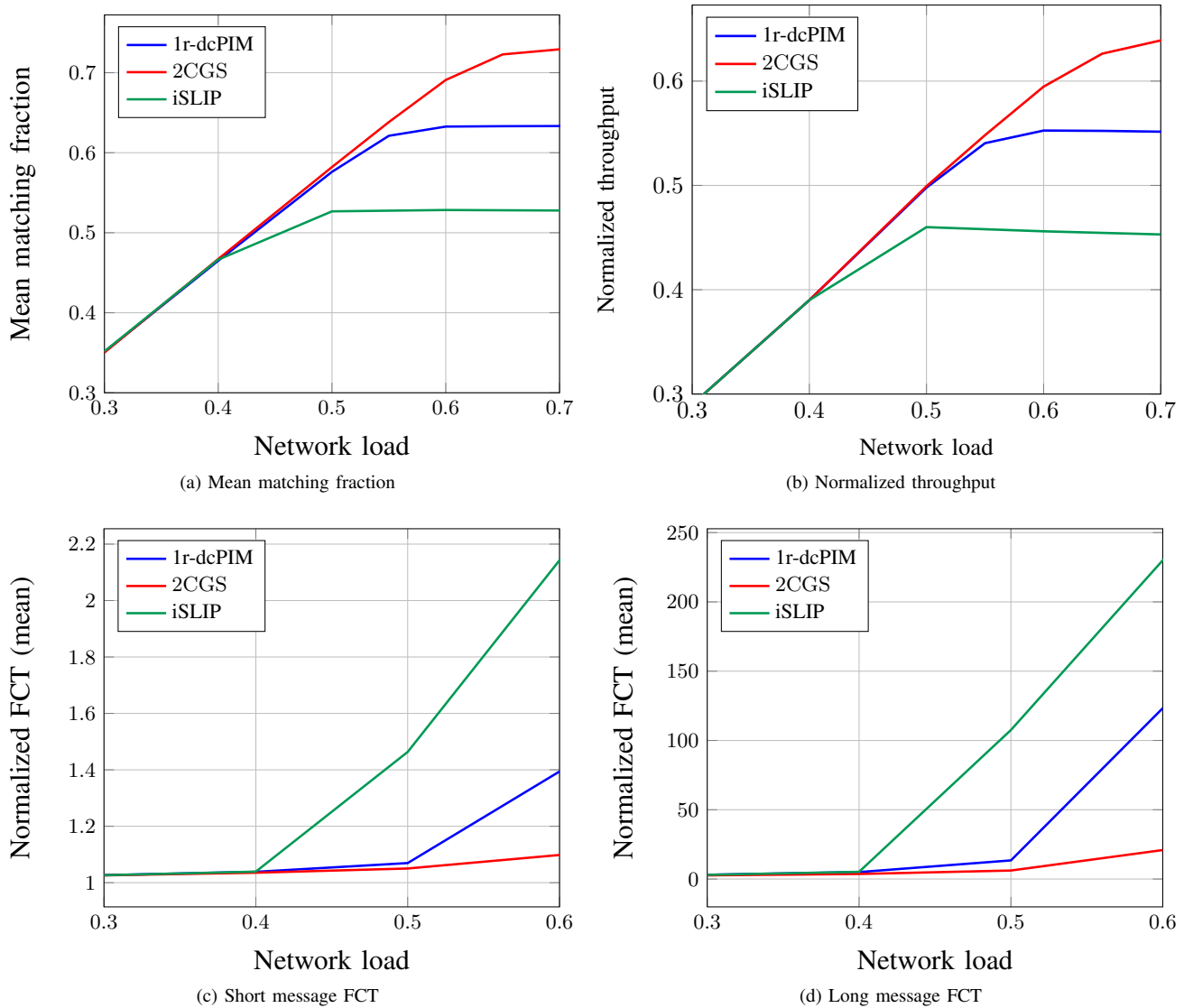
\begin{figure}[ht]
\centering
\subfloat[Mean matching fraction\label{fig:08a}]{\resizebox{0.49\linewidth}{!}{\pgfplotsset{every axis ylabel/.append style={font=\large},
	xlabel/.append style={font=\large}}
\pgfplotsset{every tick label/.append style={font=\small}}

\begin{tikzpicture}
\begin{axis}[
xlabel={Network load},
ylabel={Mean matching fraction},
xmajorgrids,
ymajorgrids,
ymin=0.3,
xmin=0.3,
xmax=0.7,
xtick ={0.3,0.4,0.5, 0.6, 0.7, 0.8, 0.9},
legend entries={ 1r-dcPIM, $\twocgs$, iSLIP
},
legend style = {legend pos = north west, nodes=right, font=\small},
title = { }
]

\addplot[
color= blue,
line width=1pt,
every mark/.append style={solid},
]
coordinates {
(0.3, 0.351111)
(0.4, 0.465000)
(0.5, 0.576104)
(0.55, 0.621160)
(0.6, 0.632663)
(0.65, 0.633154)
(0.7, 0.633334)
(0.75, 0.633315)
(0.8, 0.633411)
(0.85, 0.633223)
(0.9, 0.633581)
(0.95, 0.633230)
};

  \addplot[
color= red,
line width=1pt,
every mark/.append style={solid},
]
coordinates {
(0.3, 0.350278)
(0.4, 0.467153)
(0.5, 0.582090)
(0.55, 0.638323)
(0.6, 0.690916)
(0.65, 0.722863)
(0.7, 0.729235)
(0.75, 0.730785)
(0.8, 0.731149)
(0.85, 0.731291)
(0.9, 0.731388)
(0.95, 0.731451)
 };

  \addplot[
color= ForestGreen,
line width=1pt,
every mark/.append style={solid},
]
coordinates {
(0.3, 0.351806)
(0.4, 0.466389)
(0.5, 0.526736)
(0.6, 0.528403)
(0.7, 0.527778)
(0.8, 0.530556)
(0.85, 0.535347)
 };

\end{axis}
\end{tikzpicture}}}
\subfloat[Normalized throughput\label{fig:08b}]{\resizebox{0.49\linewidth}{!}{\begin{tikzpicture}
\begin{axis}[
    ylabel={Normalized throughput},
    xlabel={Network load},
    xmajorgrids,
    ymajorgrids,
    ymin=0.3,
    xmin=0.3,
    xmax=0.7,
%
    xtick ={0.3, 0.4, 0.5, 0.6, 0.7, 0.8, 0.9},
    legend entries={ 1r-dcPIM, $\twocgs$, iSLIP
	},
    legend style = {legend pos = north west, nodes=right, font=\small},
    title = { }
]

\addplot[
    color=blue,
    line width=1pt,
    every mark/.append style={solid},
] coordinates {
(0.3,0.29)
(0.4,0.39)
(0.50, 0.4979)
(0.55, 0.5405)
(0.60, 0.5526)
(0.65, 0.5523)
(0.70, 0.5515)
(0.75, 0.5509)
(0.80, 0.5506)
(0.85, 0.5502)
(0.90, 0.5503)
(0.95, 0.5499)
};

\addplot[
    color=red,
    line width=1pt,
    every mark/.append style={solid},
    mark indices={0,5,10,15},
] coordinates {
(0.3,0.29)
(0.4,0.39)
(0.50, 0.4992)
(0.55, 0.5482)
(0.60, 0.5947)
(0.65, 0.6263)
(0.70, 0.6390)
(0.75, 0.645)
(0.80, 0.6472)
(0.85, 0.6485)
(0.90, 0.6494)
(0.95, 0.6496)
};

\addplot[
    color=ForestGreen,
    line width=1pt,
    every mark/.append style={solid},
    mark indices={0,5,10,15},
] coordinates {
(0.3,0.29)
(0.4,0.39)
(0.50, 0.46)
(0.60, 0.456)
(0.70, 0.453)
};

\end{axis}
\end{tikzpicture}}}\\
\subfloat[Short message FCT\label{fig:08c}]{\resizebox{0.49\linewidth}{!}{\pgfplotsset{every axis ylabel/.append style={font=\large},
	xlabel/.append style={font=\large}}
\pgfplotsset{every tick label/.append style={font=\small}}

\begin{tikzpicture}
\begin{axis}[
xlabel={Network load},
ylabel={Normalized FCT (mean)},
xmajorgrids,
ymajorgrids,
xmin=0.3,
xmax=0.6,
xtick ={0.3, 0.4, 0.5, 0.6, 0.7, 0.8, 0.9},
legend entries={ 1r-dcPIM, $\twocgs$, iSLIP
},
legend style = {legend pos = north west, nodes=right, font=\small},
title = { }
]

\addplot[
color= blue,
line width=1pt,
every mark/.append style={solid},
]
coordinates {
(0.3,1.0266)
(0.4,1.0385)
(0.5,1.0695)
(0.6,1.3941)
};

  \addplot[
color= red,
line width=1pt,
every mark/.append style={solid},
]
coordinates {
(0.3,1.0259)
(0.4,1.0353)
(0.5,1.0503)
(0.6,1.0982)
 };

  \addplot[
color= ForestGreen,
line width=1pt,
every mark/.append style={solid},
]
coordinates {
(0.3,1.0264)
(0.4,1.0392)
(0.5,1.4634)
(0.6,2.1428)
(0.7,2.7182)
(0.8,3.0386)
 };

\end{axis}
\end{tikzpicture}}}
\subfloat[Long message FCT\label{fig:08d}]{\resizebox{0.49\linewidth}{!}{\pgfplotsset{every axis ylabel/.append style={font=\large},
	xlabel/.append style={font=\large}}
\pgfplotsset{every tick label/.append style={font=\small}}

\begin{tikzpicture}
\begin{axis}[
xlabel={Network load},
ylabel={Normalized FCT (mean)},
xmajorgrids,
ymajorgrids,
xmin=0.3,
xmax=0.6,
xtick ={0.3, 0.4, 0.5, 0.6, 0.7, 0.8, 0.9},
legend entries={ 1r-dcPIM, $\twocgs$, iSLIP
},
legend style = {legend pos = north west, nodes=right, font=\small},
title = { }
]

\addplot[
color= blue,
line width=1pt,
every mark/.append style={solid},
]
coordinates {
(0.3,3.06)
(0.4,5.003)
(0.5,13.44)
(0.6,123.29)
};

  \addplot[
color= red,
line width=1pt,
every mark/.append style={solid},
]
coordinates {
(0.3,2.63)
(0.4,3.64)
(0.5,6.13)
(0.6,20.91)
 };

  \addplot[
color= ForestGreen,
line width=1pt,
every mark/.append style={solid},
]
coordinates {
(0.3,3.0052)
(0.4,5.1825)
(0.5,107.53)
(0.6,230.047)
(0.7,284.27)
 };
 
\end{axis}
\end{tikzpicture}}}
\caption{Performance comparison of $\twocgs$, 1r-dcPIM, and iSLIP for SGD workload in leaf-spine architecture.}
\label{fig:08}
\end{figure}
As noted above, the stability region is precisely the range of loads for which throughput and network load are equal. 
Inspecting Figs. \ref{fig:07b} and \ref{fig:08b} suggests that the stability region extends to loads of around $0.4$--$0.5$ for iSLIP, around $0.5$ for 1r-dcPIM and around $0.6$ for $\twocgs$, but it is hard to pinpoint the precise onset of instability. 
It is easier to discern in Figs. \ref{fig:07d} and \ref{fig:08d}, where we see the flow completion time of long messages starting to increase sharply for loads above $0.4$ for iSLIP, above $0.5$ for 1r-dcPIM, and above around $0.6$ for $\twocgs$. 

In summary, all three algorithms perform well within their respective stability regions, but differ in the maximum load for which they are stable, driven by differences in the sizes of matchings that they obtain. 
It also appears from the figures that the stability region is a characteristic of the algorithm and not sensitive to the nature of the workload (IMC10 or SGD).

Figs. \ref{fig:07c} and \ref{fig:08c} show that even the FCT of short messages suffers from the onset of instability, albeit much less than for long messages (note the different scales for the $y$-axes). 
This is surprising, given that short messages are prioritized and should not see the growing long message queues. 
A possible explanation is a growth in control messages, which have even higher priority, occasioned by instability in the long message queues. 
This is supported by Fig. \ref{fig:09}, which shows the number of control messages for different choices of network loads and matching algorithms.
Indeed, the number of control messages grows sharply beyond the stability region for all three algorithms, lending support to the hypothesis that the resulting congestion delays short messages. 

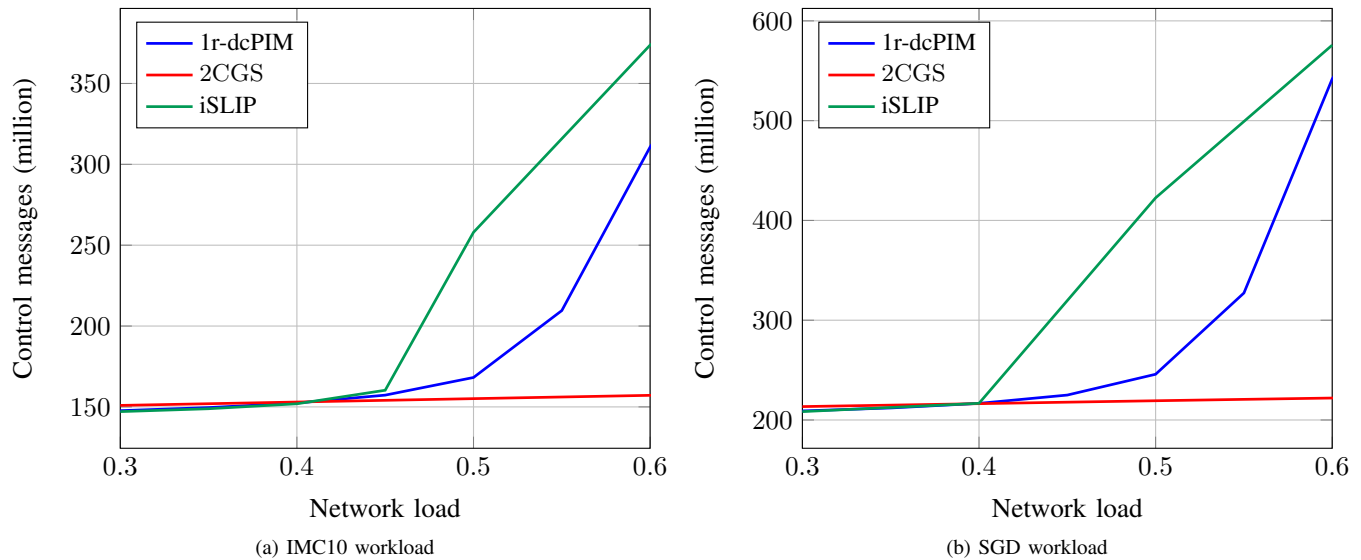
\begin{figure}[ht]
\centering
\subfloat[IMC10 workload\label{fig:09a}]{\resizebox{0.5\linewidth}{!}{\begin{tikzpicture}
\begin{axis}[
    xmin=0.3,
    xmax = 0.6,
    xtick ={0.3, 0.4, 0.5, 0.6},
    grid=both,
    minor grid style={solid,gray!50},
    legend entries={1r-dcPIM, $\twocgs$, iSLIP
	},
    legend style = {legend pos = north west, nodes=right, font=\small},
    ylabel={Control messages (million)},
    xlabel={Network load},
]

\addplot[
color= blue,
line width=1pt,
every mark/.append style={solid},
]coordinates {
(0.3,147.68)
(0.35,149.58)
(0.4,152.41)
(0.45,157.28)
(0.50, 168.19)
(0.55,209.57)
(0.60, 310.71)
(0.65,388.57)
};

\addplot[
color= red,
line width=1pt,
every mark/.append style={solid},
mark indices={0,5,10,15},
] coordinates {
(0.3,150.93)
(0.35,151.97)
(0.4,153.02)
(0.45,154.07)
(0.50, 155.11)
(0.55,156.13)
(0.60, 157.16)
(0.65,158.94)
};

\addplot[
color= ForestGreen,
line width=1pt,
every mark/.append style={solid},
mark indices={0,5,10,15},
] coordinates {
(0.3,147.09)
(0.35,148.87)
(0.4,151.93)
(0.45,160.30)
(0.50, 257.966)
(0.60, 373.73)
(0.65,409.24)
};

\end{axis}
\end{tikzpicture}}}
\subfloat[SGD workload\label{fig:09b}]{\resizebox{0.5\linewidth}{!}{\begin{tikzpicture}
\begin{axis}[
    xmin=0.3,
    xmax = 0.6,
    xtick ={0.3, 0.4, 0.5, 0.6},
    grid=both,
    minor grid style={solid,gray!50},
    legend entries={ 1r-dcPIM, $\twocgs$, iSLIP
	},
    legend style = {legend pos = north west, nodes=right, font=\small},
    ylabel={Control messages (million)},
    xlabel={Network load},
]

\addplot[
color= blue,
line width=1pt,
every mark/.append style={solid},
]coordinates {
(0.3,209.08)
(0.35,212.05)
(0.4,216.64)
(0.45,224.98)
(0.50, 245.83)
(0.55,327.26)
(0.60, 541.84)
(0.65,674.46)
};

\addplot[
color= red,
line width=1pt,
every mark/.append style={solid},
mark indices={0,5,10,15},
] coordinates {
(0.3,213.43)
(0.35,214.89)
(0.4,216.36)
(0.45,217.85)
(0.50, 219.28)
(0.55,220.67)
(0.60, 222)
(0.65,225)
};

\addplot[
color= ForestGreen,
line width=1pt,
every mark/.append style={solid},
mark indices={0,5,10,15},
] coordinates {
(0.3,208.39)
(0.4,216.70)
(0.50, 422.64)
(0.60, 575.73)
};

\end{axis}
\end{tikzpicture}}}
\caption{
Total number of control messages generated by $\twocgs$, 1r-dcPIM, and iSLIP for the IMC10 and SGD workloads under different network loads in a leaf-spine architecture.
}
\label{fig:09}
\end{figure}


\section{Conclusion}
%
The work presented here is motivated by the challenge of scheduling latency-sensitive communications between hosts in large data center networks, where centralized coordination is infeasible due to scale.
We abstracted this scheduling problem as one of decentralized bipartite matching and proposed a family of single-round probabilistic algorithms built on two complementary ideas: random thinning of the communication graph and degree-biased receiver selection using only local information.

On the theoretical side, we established closed-form asymptotic expressions for the mean matching fraction under both uniform selection, DB$(0)$, and greedy selection, DB$(-\infty)$, for $D$-out random graph models.
A notable finding is that the mean matching fraction under uniform selection depends on the out-degree distribution only through the probability that a sender has at least one feasible receiver, making it insensitive to the shape of the degree distribution.
Under greedy selection, targeting minimum-degree receivers yields gains in sparse regimes but suffers from a bottleneck effect as density increases: low-degree receivers attract disproportionately many grants, leading to collisions.
Random thinning resolves this tension by sparsifying dense graphs before matching.

Numerical evaluations on finite systems confirmed the tightness of the theoretical predictions and showed that thinning to a target degree of two, combined with greedy selection (the 2CGS algorithm) delivers robust performance across a wide range of graph models and parameters without requiring any tuning.
System-level packet simulations using production workload traces demonstrated that 2CGS extends the stability region to approximately $60\%$ network load, compared with roughly $50\%$ for one-round dcPIM and $40$--$50\%$ for iSLIP, while incurring comparable control-message overhead within the stable operating regime.
These results make 2CGS a promising candidate for further investigation in hardware deployment.

\bibliographystyle{IEEEtran}
\bibliography{IEEEabrv,tnet-2026}
\end{document}